\documentclass[aps,prd,onecolumn,nofootinbib]{revtex4-2}

\usepackage[utf8]{inputenc}
\usepackage[T1]{fontenc}
\usepackage{lmodern}
\usepackage{amsmath,amssymb}
\usepackage{physics}
\usepackage{graphicx}
\usepackage{hyperref}
\usepackage{xcolor}
\usepackage{bm}
\usepackage{microtype}
\usepackage{comment}

\usepackage[normalem]{ulem}

\usepackage{amsmath}
\usepackage{amsfonts}
\usepackage{amssymb}
\usepackage{graphicx}%
\usepackage{orcidlink}

\providecommand{\U}[1]{\protect\rule{.1in}{.1in}}
\newtheorem{theorem}{Theorem}

\newtheorem{claim}[theorem]{Claim}

\newtheorem{definition}[theorem]{Definition}

\newtheorem{lemma}[theorem]{Lemma}

\newenvironment{proof}[1][Proof]{\noindent\textbf{#1.} }{\ \rule{0.5em}{0.5em}}

\begin{document}

\title{A  strongly hyperbolic viscous relativistic hydrodynamics theory with first-order charge current}

% --- Autor 1 ---
\author{Federico Schianchi\protect\orcidlink{0000-0001-7646-5988}}
\thanks{Email: \href{mailto:federico.schianchi@uib.cat}{federico.schianchi@uni-jena.de}}

\affiliation{Departament de Física \& Institute of
Applied Computing and Community Code (IAC3), Universitat de les Illes Balears, Palma de Mallorca, Baleares E-07122, Spain}
%
% --- Autor 2 ---
\author{Fernando Abalos\protect\orcidlink{0000-0001-7863-3711}}
\thanks{Email: \href{mailto:j.abalos@uib.es}{j.abalos@uib.es}} % -> \dagger

\affiliation{Departament de Física \& Institute of
Applied Computing and Community Code (IAC3), Universitat de les Illes Balears, Palma de Mallorca, Baleares E-07122, Spain}

\date{\today}

\begin{abstract}
We extend the first order dissipative relativistic hydrodynamics model of Bemfica-Disconzi-Noronha-Kovtun (BDNK)  in order to include the charge number current in full first order expansion with out-of-equilibrium contribution proportional to the evolution equation of the ideal fluid. We obtain a fully second order system of partial differential equations (PDEs) that can be casted in a fully conservative way. We analyze the hyperbolicity of this model coupled to Einstein field equations using a newly developed technique that allows for hyperbolicity studies without explicit first order reduction. Furthermore, we identify a frame choice where our formulation is causal, stable and with positive entropy generation for a wide range of equations of state (EoS). Our analysis shows that the inclusion of an out-of-equilibrium correction to the charge current, plays an important role in guaranteeing the strong hyperbolicity and, therefore, the well-posedness of the system. 
\end{abstract}

\maketitle

\section{Introduction}
Relativistic hydrodynamics plays a fundamental role in several branches of physics. However, most of the work produced so far has been limited to the ideal fluid case, neglecting any dissipative effect. Dissipative effects have been suggested to play a role in different applications, from nuclear physics (heavy ion collision \cite{Gale:2013da,Romatschke:2017ejr}), to astrophysics (compact objects accretion \cite{Abramowicz:2011xu,MoeenMoghaddas:2016dng}, supernovae, binary neutron star mergers \cite{Chabanov:2023abq,Alford:2017rxf,Most:2021zvc,Shibata:2017jyf,Chabanov:2023blf}) and cosmology. Attempts to go beyond the ideal fluid started in the last century by Eckart \cite{Eckart:1940te} and Landau and Lifshitz \cite{landau1987fluid}. However, the problem remained still open, as both these theories resulted to be acausal and unstable. This was due to the parabolic nature of the equations of motions they were giving rise to.
The issue was later fixed by the M\"{u}ller-Israel-Stewart (MIS) formulation \cite{Muller:1967zza,Israel:1976tn,Israel:1979wp}. In this theory, the dissipative quantities (bulk viscosity and heat conduction, etc.) are promoted to additional degrees of freedom, controlled by second order evolution equations that make them relax to the values predicted by the theories of Eckart and Landau-Lifshiz. Relaxation is then controlled by specific time parameters introduced ad hoc. The result is a strongly hyperbolic (and therefore well-posed), causal and stable theory. Despite being a very phenomenological model, MIS constituted the state of the art in dissipative relativistic hydrodynamics simulations for decades.
Unfortunately, in the 1990s, the MIS formulation was proven to encounter difficulties in presence of strong shocks \cite{Olson:1990pnz,Geroch:1991}. Moreover, its need for extra evolved variables and stiff relaxation equations introduces significant difficulties in numerical simulations. 

In the last years, a new effective-field-theory (EFT) for relativistic dissipative hydrodynamics close to equilibrium has been proposed by Bemfica, Disconzi, Noronha and Kovtun (BDNK) \cite{Bemfica:2017wps,Kovtun:2019hdm,Bemfica:2019knx,Bemfica:2020zjp}. This theory is still based on the pure conservation of stress-energy tensor and charge current as only evolution equations, without the need of adding extra evolved variables. Dissipative effects are rather included as first order corrections in the fundamental variable gradients to the ideal conserved tensors, giving rise to second order evolution equations. If we consider the most general first-order gradient expansion, we can see that both Eckart and Landau-Lifshitz theories can be obtained from BDNK as a particular choice of expansion parameters, also called choice of frame. 

In the BDNK theory the stress-energy tensor and charge current are expressed as
\begin{align}
    T^{\mu\nu} &= (\epsilon+\mathcal{A}) u^\mu u^\nu + (P+\Pi)\Delta^{\mu\nu} + u^\mu Q^\nu + u^\nu Q^\mu + \mathcal{T}^{\mu\nu} \label{eq:T_1}\\
    J^\mu &= (n+N) u^\mu + \mathcal{J}^\mu \label{eq:J_1}
\end{align}
where $u^\mu$ is the normalized four-velocity of the fluid, satisfying $u^\mu u_\mu = -1$, and $Q^\mu u_\mu = \mathcal{J}^\mu u_\mu = \mathcal{T}^{\mu\nu}u_\mu=0$. The tensor $\Delta^{\mu\nu} \equiv g^{\mu\nu} + u^\mu u^\nu$ acts as the spatial projector orthogonal to $u^\mu$, and $\epsilon$, $n$, and $P$ denote the energy density, charge number density, and pressure of the fluid at local thermal equilibrium, respectively, i.e. the ones of an ideal fluid.

The non-ideal terms $\mathcal{A}$, $\Pi$, $Q^{\mu}$, $\mathcal{T}^{\mu\nu}$, $N$ and $\mathcal{J^{\mu}}$ are expressed as first order expansions in the hydrodynamics gradients in the most general way by the  constitutive relations  (see \cite{Kovtun:2012rj} for a detailed exposition):

\begin{align}
    \mathcal{A} &= \epsilon_1 \frac{u^\alpha \nabla_\alpha T}{T} + \epsilon_2 \nabla_\alpha u^\alpha + \epsilon_3 u^\alpha \nabla_\alpha \left( \frac{\mu}{T} \right) + \mathcal{O}(\nabla^2) \label{eq:constitutive_relation_1}\\
    \Pi &= \pi_1 \frac{u^\alpha \nabla_\alpha T}{T} + \pi_2 \nabla_\alpha u^\alpha + \pi_3 u^\alpha \nabla_\alpha \left( \frac{\mu}{T} \right) + \mathcal{O}(\nabla^2) \label{eq:constitutive_relation_2} \\
    Q^\mu &= \theta_1 \frac{\Delta^{\mu\alpha} \nabla_\alpha T}{T} + \theta_2 u^\alpha \nabla_\alpha u^\mu + \theta_3 \Delta^{\mu\alpha} \nabla_\alpha \left( \frac{\mu}{T} \right) + \mathcal{O}(\nabla^2) \label{eq:constitutive_relation_3}\\
    \mathcal{T}^{\mu\nu} &= -2\eta \sigma^{\mu\nu} + \mathcal{O}(\nabla^2) \label{eq:constitutive_relation_4} \\
    N &= \nu_1 \frac{u^\alpha \nabla_\alpha T}{T} + \nu_2 \nabla_\alpha u^\alpha + \nu_3 u^\alpha \nabla_\alpha \left( \frac{\mu}{T} \right) + \mathcal{O}(\nabla^2) \label{eq:constitutive_relation_5} \\
    \mathcal{J}^\mu &= \gamma_1 \frac{\Delta^{\mu\alpha} \nabla_\alpha T}{T} + \gamma_2 u^\alpha \nabla_\alpha u^\mu + \gamma_3 \Delta^{\mu\alpha} \nabla_\alpha \left( \frac{\mu}{T} \right) + \mathcal{O}(\nabla^2) \label{eq:constitutive_relation_6}
\end{align}
where $\epsilon_i$, $\pi_i$, $\theta_i$, $\nu_i$ and $\gamma_i$ are called constitutive parameters, $\eta$ is the shear viscosity, $T$ the fluid`s temperature and $\mu$ the charge chemical potential. In general, the constitutive parameters, as well as $\epsilon$, $n$ and $P$, are functions of the fundamental variables $T$, $\mu$. Moreover, we define the shear tensor as
\begin{equation}
    \sigma^{\mu \nu} \equiv \left( \Delta^{\mu (\alpha} \Delta^{\beta)\nu} - \frac{1}{3} \Delta^{\mu\nu} \Delta^{\alpha \beta} \right) \nabla_\alpha u_\beta \label{Eq:sigma_ab}
\end{equation}
With this ansatz and the constitutive relations Eqs. \eqref{eq:constitutive_relation_1}-\eqref{eq:constitutive_relation_6}, Eqs. \eqref{eq:T_1} and \eqref{eq:J_1} provide the most general expression of stress-energy tensor and charge current at first order expansion in gradients.

When the fluid is out of equilibrium, the quantities $u^{\mu}$, $T$, and $\mu$ are no longer determined by microscopic definitions. Instead, they should be regarded as auxiliary hydrodynamic variables used to parameterize the physical observables $T^{\mu\nu}$ and $J^{\mu}$.
In this situation, there is no unique or preferred choice for the set of variables $v \equiv (T, \mu, u^{\mu})$. Given one such set, an alternative definition 
$v' \equiv (T', \mu', u'^{\mu})$ 
can be introduced as
\begin{align}
T' &= T + a_1 \frac{u^\lambda \nabla_\lambda T}{T} + a_2 \nabla_{\lambda} u^{\lambda} + a_3 u^\lambda \nabla_{\lambda} \left( \frac{\mu}{T} \right) + O(\nabla^2), \label{eq:transform_T}\\
u'^{\mu} &= u^{\mu} + b_1\frac{\Delta^{\mu\nu} \nabla_{\nu} T}{T} +  b_2 u^\lambda \nabla_\lambda u^{\mu} + b_3 \Delta^{\mu\nu} \nabla_{\nu} \left( \frac{\mu}{T} \right) + O(\nabla^2), \label{eq:transform_u} \\
\mu' &= \mu + c_1 \frac{u^\lambda \nabla_\lambda T}{T} + c_2 \nabla_{\lambda} u^{\lambda} + c_3 u^\lambda \nabla_{\lambda} \left(  \frac{\mu}{T} \right) + O(\nabla^2), \label{eq:transform_mu}
\end{align}
where $a_i$, $b_i$, and $c_i$ are functions of the hydrodynamic variables $v$. 
This represents the most general covariant redefinition of the variables up to first order in derivatives.

At this point we impose that the physical observables $T^{\mu\nu}$ and $J^\mu$ remain unchanged by the transformation up to first order in the gradients:
\begin{align}
T'^{\mu\nu}(T', \mu', u'^{\mu}) &= T^{\mu\nu}(T, \mu, u^{\mu}) + O(\nabla^2), \\
J'^{\mu}(T', \mu', u'^{\mu}) &= J^{\mu}(T, \mu, u^{\mu}) + O(\nabla^2).
\end{align}
One can show that, in order to satisfy these conditions, we have to perform a transformation on the constitutive parameters defined by:
\begin{align}
\epsilon_i' (T, \mu) &= \epsilon_i (T, \mu) - \frac{\partial \epsilon}{\partial T} a_i - \frac{\partial \epsilon}{\partial \mu} c_i, \label{eq:frame_transformation_1}\\
\pi_i' (T, \mu) &= \pi_i (T, \mu) - \frac{\partial P}{\partial T} a_i - \frac{\partial P}{\partial \mu} c_i, \label{eq:frame_transformation_2}\\
\nu_i' (T, \mu) &= \nu_i (T, \mu) - \frac{\partial n}{\partial T} a_i - \frac{\partial n}{\partial \mu} c_i, \label{eq:frame_transformation_3}\\
\theta_i'(T, \mu) &= \theta_i (T, \mu) - (\epsilon + P) b_i, \label{eq:frame_transformation_4}\\
\gamma_i'(T, \mu) &= \gamma_i(T, \mu)- n b_i, \label{eq:frame_transformation_5}\\
\eta'(T, \mu)  &= \eta (T, \mu), \label{eq:frame_transformation_6}
\end{align}
Where we omitted the dependence of $n$, $\epsilon$, $P$, $a_i$, $b_i$ and $c_i$  on $T$ and $\mu$ for notational simplicity.
This shows that the choice of $(T, \mu, u^{\mu})$ is merely a convenient parametrization of the physical observables, valid up to second order in derivatives (see \cite{Kovtun:2012rj}).
The parameters $(a_i, b_i, c_i)$ characterize the frame transformation. Since their choice modifies the constitutive parameters $(\epsilon_i', \pi_i', \nu_i', \theta_i', \gamma_i')$, which in turn determine the quantities $\mathcal{A}$, $\Pi$, $Q^{\mu}$, $\sigma^{\mu\nu}$, $N$, and $\mathcal{J}^{\mu}$ (see Eqs. \eqref{eq:constitutive_relation_1}-\eqref{eq:constitutive_relation_6}), we identify the frame choice with the specification of the constitutive parameters. 
As an example, the Landau frame corresponds to the choice where the heat flux in the local rest frame vanishes ($\theta_i'=0$), this corresponds to a particular choice of fluid's velocity that can always be chosen with a transformation of $u^\mu$ defined by a tetrad of parameters $b_1$, $b_2$, $b_3$.

On the other hand, certain combinations of the constitutive parameters are frame invariant, meaning that they remain unchanged under any frame transformation. They are:
\begin{align}
    f_i &= \pi_i - \left.\frac{ \partial P}{ \partial \epsilon} \right|_n \epsilon_i - \left. \frac{\partial P}{ \partial n} \right|_\epsilon \nu_i, \label{eq:frame_inv_1}\\
    \ell_i &= \gamma_i - \frac{n}{\epsilon+P} \theta_i. \label{eq:frame_inv_2}
\end{align}
These invariants will be used to identify physical quantities, and we will return to them later.

In \cite{Kovtun:2019hdm,Hoult:2020eho}, the BDNK theory has been proven to be causal and stable for a suitable choice of the constitutive parameters. Later in \cite{Bemfica:2017wps,Bemfica:2020zjp}, the model is written using $\epsilon$ and $n$ as dynamic variables instead of $T$ and $\mu$. This choice, even if less direct from the thermodynamical point of view, makes the formulation more suitable for simulation purposes in astrophysics and high energy physics contexts. Moreover it has been proven to be strongly hyperbolic, causal, stable and with positive entropy generation, always under certain conditions on the constitutive parameters. In a dissipative fluid theory, strong hyperbolicity is essential to guarantee a well-posed initial-value problem (and hence the predictibility power of the theory, together with robust numerical evolutions, see \cite{KreLor89,GusKreOli95,SarTig12,Hilditch:2013sba,Aba17} for more details), causality is required to ensure finite signal speeds consistent with relativity, and stability together with positive entropy production provides basic physical consistency by ruling out unphysical growing modes and enforcing the second law of thermodynamics.

The first numerical tests have been performed, showing the theory is solvable with standard numerical techniques as explicit time integration with finite volumes/differences \cite{Pandya:2021ief,Pandya:2022sff,Bea:2023rru,Bhambure:2024axa,Fantini:2025gnm,Keeble:2025bkc,Shum:2025jnl,Bea:2025eov}, and the constraints coming from first order reduction are stable and propagated correctly \cite{Fantini:2025gnm}. Recently, a fully first order flux-conservative reduction has been developed \cite{Clarisse:2025lli}, providing the mathematical tools for correctly capturing discontinuous solutions . However, all the tests performed so far employed a simplified equation of state (EoS) that does not depend on the charge number, with the exception of \cite{Pandya:2022sff}, which employed an ideas gas EoS, but only in one-dimensional simulation domains.
Additionally, it has been proven that this theory can be derived from perturbative kinetic theory employing the Boltzmann equation \cite{Rocha:2022ind}. In \cite{Mullins:2023ott,Gavassino:2024vyu} it is shown that including fluctuations theory to first order hydrodynamics inevitably leads to the introduction of non-local constraints. However, this does not necessarily indicate a fundamental issue at the level of physical observables. Rather, in \cite{Dore:2021xqq,Sampaio:2025jtp}, it is show that such constraints provide guidance for ensuring compatibility between the hydrodynamic description and equilibrium statistical mechanics.

In this work we will focus on a formulation of BDNK in terms of $\epsilon$, $n$ and $u^\mu$, taking the model of \cite{Bemfica:2020zjp} as a starting point. 
%%%%%%%%%%%%%%%%%%%%%%%%%%%%%%%%%%%%%%%%5

In Ref.~\cite{Bemfica:2020zjp}, the evolution system has mixed differential
order. The equations for $\epsilon$ and $u^\mu$ are second order, while the
charge sector is described by the first-order conservation law 
\begin{equation}
    \nabla_\nu(nu^\nu)=0 . \label{current_1}
\end{equation}

The hyperbolicity analysis is then carried out after reducing the second-order
equations to first order with a constraint-adapted reduction, transforming the first order charge conservation equation into a constrain equation. The constraint-adapted order reduction is designed to impose Eq. \eqref{current_1}, allowing to eliminate $u^\mu \nabla_\mu n$ from the system, removing one degree of freedom. Furthermore, the normalization condition $u^\mu u_\mu=-1$ is not imposed algebraically, but propagated by evolution once imposed on initial data. Hence, \cite{Bemfica:2020zjp} establishes well posedness on the subset where this condition is satisfied, which is the one of physical interest. This model, is then proven equivalent (on constraint satisfying initial data) to the one where the normalization of $u^\mu$ is imposed by construction, proving the well-posedness of the latter. 

In this article we propose a modification of the model developed in \cite{Bemfica:2020zjp} where the charge evolution equation is elevated to second order by adding non-ideal terms to $J^\mu$ in the same fashion used for $T^{\mu\nu}$. This allows us to investigate the hyperbolicity of the model using the tools developed in \cite{Abalos:2026kim}. Additionally, we algebraically impose the constraint $u^\mu u_\mu=-1$ throughout the article, unless differently specified. We see this choice as more a straightforward approach, as it allows to work directly with the set of variables that would be employed in a numerical simulation, avoiding the chain of different equivalent models used in \cite{Bemfica:2020zjp}. Numerical evolution, indeed, requires the model to be well-posed on the whole ambient Sobolev space, not only on a constrained subsection of it. The reason is that truncation error would lead to a violation of the constraints during evolution, potentially bringing the solution in a region of the Sobolev space where well-posedness is not guaranteed. In terms of our model, this means that we can rather impose the normalization of $u^\mu$ algebraically and only evolve three of its components, or show that our model is well-posed independently on the normalization of $u^\mu$. We chose the first approach as it simplyfies the calculations reducing the degrees of freedom. In Appendix \ref{appendix:non_normalized_u_hyperbolicity}, however, we show how the hyperbolicity of our model is preserved also when extended to the non-normalized fluid's 4-velocity (with Eq. \eqref{eq:evolution3} being considered as 4 evolution equations instead of 3), due to its lack of degenerate zero-velocity modes.

Stability and causality of fully first order charged fluids has been already investigated in \cite{AbboudSperanzaNoronha2024}, however, such article does not present an  hyperbolicity analysis. Moreover, the conditions found in \cite{AbboudSperanzaNoronha2024} are expressed in terms of completely generic transport coefficients, while in this work we restrict to the on-shell EFT constitutive-relations, namely choices that make the non-equilibrium terms proportional to the evolution equation of the ideal fluid, as done in \cite{Bemfica:2020zjp}. Our BDNK formulation is presented in the following section, where all conditions for causality, stability and entropy production are discussed.

One could be tempted to just keep a first-order charge advection equation as in the ideal fluid scenario, and just evolve a mixed system of first and second order equations. Despite this being mathematically possible, it would introduce two challenges. First, when solving the system numerically, we could not apply a conservative-to-primitive scheme, as the one used in \cite{Shum:2025jnl}, since the linear algebraic system defining the variable mapping is underdetermined for this formulation. Second, we could not directly study its hyperbolicity using the definition in \cite{Abalos:2026kim}. Instead, one must first reduce the system to first order, via differential or pseudo-differential reduction techniques, and then apply a different but equivalent definition (see \cite{Reu04,NagOrtReu04,GunGar05,Abalos:2026kim}). Such reductions increase the algebraic complexity and, since it is a different system to the study in this paper, may lead to different restrictions on the constitutive parameters.
This motivates our use of the definition of \cite{Abalos:2026kim}, developed within the framework of matrix pencil theory, which is particularly useful for our hydrodynamic formulation. It enables the analysis of hyperbolicity without an explicit first-order reduction, greatly simplifying the calculations, as we will show below. More generally, matrix pencil techniques also play an important role in the study of hyperbolicity for PDEs with differential constraints, for further details we refer the reader to \cite{AbaReu18,Aba22,AbaReuHil24}. One advantage of our system is there is no zero characteristic velocities, all characterisitc modes propagate as waves. The zero velocities in general can bring numerical problems at the boundaries in numerical simulations.

It is important to clarify that, despite appearances, our formulation is still not fully conservative. Gradient-based constitutive relations, indeed, lead to derivatives of the fluid variables in the source (non-principal) terms. A fully conservative formulation would require a first order reduction in both space and time as done in \cite{Shum:2025jnl,Clarisse:2025lli}. However, such reduction preserves the hyperbolicity properties, see  \cite{GunGar05} for more details. In addition, the reduction to first order introduces differential constraints that should be preserved by evolution, this preservation is guaranteed as it is explained in \cite{GunGar05}. In this article we do not treat this issue, but constraint advection and damping techniques can be employed to enhance the accuracy in numerical simulations, as done in \cite{Clarisse:2025lli}. The so called trace-fixed particle frame model presented in \cite{Salazar:2024wly} also provides a strongly hyperbolic, fully first order, conservative model for charged fluids. However this comes at the cost of losing the frame invariance. In other words, this model constitute a very specific frame choice of the constitutive relations \eqref{eq:constitutive_relation_1}-\eqref{eq:constitutive_relation_6}. The model we are proposing in this article, on the contrary, keeps the freedom of specifying several frame parameters.

This paper is structured as follows. In section \ref{section:fluid_model} we introduce our model, identical to \cite{Bemfica:2020zjp} in the stress-energy tensor, but different in the charge current. We study its causality, hyperbolicity, stability and entropy production. Finally, we obtain a frame where all the desired properties are satisfied. In Section \ref{section:coupling_with_gravity} we show how the hyperbolicity properties are preserved also when the system is coupled to Einstein field equations. In Appendix \ref{appendix:non_normalized_u_hyperbolicity} we show how the hyperbolicity of our model is preserved also without assuming the normalization of the fluid's 4-velocity $u^\mu$. In addition, we show the constraint preservation of the normalization $u^\alpha$ and its importance to guarantee the equivalence of the systems.  In Appendix \ref{appendix:betaniszero} we investigate the case $\beta_n=0$, which makes the $T^{\mu\nu}$ and $J^\mu$ sectors of the hydrodynamics evolution equations weakly coupled, and where the strong hyperbolicity discussion is simplified.  We will use natural units throughout this article, namely, $c=G=M_\odot=1$.

\section{New  Hydrodynamic formulation, hyperbolicity, causality, entropy and stability} \label{section:fluid_model}

In this section, we introduce the family of first-order corrections to the stress--energy tensor and the charge current considered in this work. 
We explain the motivation for this choice and state the conditions ensuring \textit{strong hyperbolicity} and \textit{causality} (Lemma~\ref{Lemma:hip:cau}), \textit{non negative entropy production} (Lemma~\ref{Lemma:entropy}), and finally \textit{linear stability} (Lemma~\ref{stability}). 
Since these conditions involve a large number of inequalities, their analytical factorization is extremely challenging, if not impossible. 
For this reason, in Subsection~\ref{subsection:choice_of_frame}, we introduce a physically motivated frame parametrization and present Claim~\ref{claim}. 
This claim provides numerical ranges for the parameters that satisfy all the conditions in Lemmas~\ref{Lemma:hip:cau}, \ref{Lemma:entropy}, and \ref{stability}, showing the existence of solutions fulfilling all requirements and providing a basis for a possible numerical implementation of our formulation.

\subsection{Setting and evolution equations}

We describe the system in terms of $\epsilon$, $n$ and $u^\mu$ as fundamental variables. Our approach for $J^\mu$ and $T^{\mu\nu}$ include the following particular choice for the first-order terms in the gradient expansion 
\begin{align}
    \mathcal{A} &= \tau_\varepsilon \left[ u^\lambda\nabla_\lambda\epsilon + (\epsilon+P)\nabla_\lambda u^\lambda \right] + \mathcal{O}(\nabla^2) \label{eq:A}\\
    \Pi &= -\zeta \nabla_\lambda u^\lambda + \tau_P \left[ u^\lambda \nabla_\lambda\epsilon + (\epsilon+P)\nabla_\lambda u^\lambda \right] + \mathcal{O}(\nabla^2) \label{eq:Pi} \\
    Q^\mu &= \tau_Q(\epsilon+P) u^\lambda \nabla_\lambda u^\mu + \beta_\epsilon \Delta^{\mu\lambda} \nabla_\lambda \epsilon + \beta_n \Delta^{\mu\lambda} \nabla_\lambda n + \mathcal{O}(\nabla^2) \label{eq:Q}\\
    \mathcal{T}^{\mu\nu} &= -2\eta \sigma^{\mu\nu} + \mathcal{O}(\nabla^2)  \\
    N &= \tau_n (n \nabla_\mu u^\mu + u^\mu\nabla_\mu n) + \mathcal{O}(\nabla^2)  \label{eq:N}\\
    \mathcal{J}^\mu &= \tau_\mathcal{J} n u^\lambda \nabla_\lambda u^\mu + \lambda_\epsilon \Delta^{\mu \nu} \nabla_\nu \epsilon + \lambda_n \Delta^{\mu \nu} \nabla_\nu n + \mathcal{O}(\nabla^2) \label{eq:J}
\end{align}
where $\mathcal{A}$, $\Pi$ and $Q^\mu$ are the same as in \cite{Bemfica:2020zjp}. Here $\zeta$ is the bulk viscosity, and we have a set of parameters $\tau_\epsilon, \tau_Q, \tau_P, \tau_n, \tau_\mathcal{J}, \beta_\epsilon, \beta_n,  \lambda_\epsilon, \lambda_n$ that can be freely chosen. They are in general function of the main variables of the system. We restrict ourselves to $\lambda_{n,\epsilon}$ and $\beta_{n,\epsilon}$ defined as 
\begin{align}
    \lambda_\epsilon &\equiv -\sigma_0 T \left. \frac{\partial(\mu/T)}{\partial \epsilon} \right|_n + \tau_\mathcal{J} \frac{n}{\epsilon+P}\left. \frac{\partial P}{\partial \epsilon} \right|_n \label{lambda_1} \\
    \lambda_n &\equiv -\sigma_0 T \left. \frac{\partial(\mu/T)}{\partial n} \right|_\epsilon + \tau_\mathcal{J}  \frac{n}{\epsilon+P}\left. \frac{\partial P}{\partial n} \right|_\epsilon \label{lambda_2} \\
    \beta_\epsilon &\equiv \sigma T \frac{\epsilon+P}{n} \left. \frac{\partial(\mu/T)}{\partial \epsilon} \right|_n + \tau_Q \left. \frac{\partial P}{\partial \epsilon} \right|_n \label{beta_1}\\
    \beta_n &\equiv \sigma T \frac{\epsilon+P}{n} \left. \frac{\partial(\mu/T)}{\partial n} \right|_\epsilon + \tau_Q \left. \frac{\partial P}{\partial n} \right|_\epsilon \label{beta_2}
\end{align}
where $\sigma_0$ and $\sigma$ are the charge number diffusivity and the heat conductivity, respectively. Here the derivatives with respect to $n$ and $\epsilon$ are taken at thermodynamical equilibrium and they are determined by the equation of state of the fluid. 
This choice of $\lambda_{n,\epsilon}$ and $\beta_{n,\epsilon}$ is motivated by the fact we want $Q^\mu$ and $\mathcal{J}^\mu$ to satisfy the following expressions 
\begin{align}
    \mathcal{J}^\mu &= -\sigma_0 T \Delta^{\mu \lambda}\nabla_\lambda\left( \frac{\mu}{T}\right) + \tau_\mathcal{J}  \frac{n}{\epsilon+P}[(\epsilon+P) u^\lambda \nabla_\lambda u^\mu + \Delta^{\mu\lambda}\nabla_\lambda P] + \mathcal{O} (\nabla^2)\\
    Q^{\mu} &= \sigma T \frac{(\epsilon + P)}{n} 
\Delta^{\mu\lambda} \nabla_{\lambda} \left( \frac{\mu}{T} \right)
+ \tau_{Q} \left[ (\epsilon + P) u^{\lambda} \nabla_{\lambda} u^{\mu}
+ \Delta^{\mu\lambda} \nabla_{\lambda} P \right] + \mathcal{O} (\nabla^2)  \label{Eq:Q_1}  
\end{align}
There is one key difference between the $Q^\mu$ of this formulation and the one of \cite{Bemfica:2020zjp}. In this formulation $Q^\mu$ represents the total flux of energy normal to $u^\mu$, hence, it contains both the energy diffused by pure heat conduction (which follows the Eckart law), and the energy advected by particle's diffusion (proportional to $\mathcal{J}^\mu$). Since both terms are proportional to $\nabla_\lambda(\frac{\mu}{T})$, they can be factorized into a single coefficient $\sigma=\sigma_E-\sigma_0$, with $\sigma_E>0$ being the heat conductivity of the Eckart law. Notice that in this model, it is possible to have $\sigma<0$ when $Q^\mu$ is dominated by $\sigma_0$, i.e. the energy advected by particles.

The terms $\nabla_\lambda\left( \frac{\mu}{T}\right)$ and $\nabla_\lambda P$ can be rewritten as derivatives of the fundamental variables $\epsilon$ and $n$ if needed. To completely determine the evolution equations, an equation of state should be prescribed i.e. the expression for $P$, $\mu$ and $T$ as a function of the fundamental variables $(\epsilon, n)$.  Then the first order corrections terms $\mathcal{A}, \Pi, Q^\mu, \mathcal{T}^{\mu\nu}, N, \mathcal{J}^\mu$  are described by a set of nine parameters, three of them are physical $\eta, \zeta, \sigma + \sigma_0$ i.e. their value comes from some physical measurement at equilibrium, and five of them $\tau_\epsilon, \tau_Q, \tau_P, \tau_n, \tau_\mathcal{J}$ are 
 frame parameters, since they define the choice of frame.
 One can verify that $\eta$, $\zeta$, and $\sigma+\sigma_0$ are physical quantities, since they can be written in terms of the frame invariants \eqref{eq:frame_inv_1}-\eqref{eq:frame_inv_2} and the main variables, and therefore their functional form is invariant under frame transformations. In particular, for $\sigma+\sigma_0$ one finds
\begin{equation}
    \sigma+\sigma_0= \frac{n}{\epsilon+P}\ell_1 - \frac{1}{T} \ell_3.
\end{equation}
On the other hand, we notice that the most general first order corrections \eqref{eq:T_1}-\eqref{eq:constitutive_relation_6} include 16 parameters, so we  are considering a subfamily of this general one.

Finally, the set of evolution equations follow from the  conservation equations 
\begin{align}
    \nabla_\mu T^{\mu\nu} &= 0, \label{Eq:D_T} \\
    \nabla_\mu J^{\mu} &=0. \label{Eq:D_J}
\end{align}
They are a set of five second order evolution equations for our five main variables $(\epsilon, n, u^\mu)$, assuming $u^\mu$ normalized. They can be written as 
\begin{align}
    \nabla_\mu J^\mu &= n \nabla_\mu u^\mu + u^\mu \nabla_\mu n + N \nabla_\mu u^\mu + u^\mu \nabla_\mu N + \nabla_\mu \mathcal{J}^\mu = 0, \label{eq:evolution1} \\
    -u_\nu \nabla_\mu T^{\mu\nu} &= u^\lambda \nabla_\lambda \epsilon + (\epsilon+\mathcal{A}+P+\Pi)\nabla_\lambda u^\lambda + u^\lambda \nabla_\lambda \mathcal{A} + \nabla_\lambda Q^\lambda + Q^\nu u^\lambda \nabla_\lambda u_\nu + 2\eta \sigma^{\mu\nu}\sigma_{\mu\nu} = 0, \label{eq:evolution2} \\
    \Delta^\beta_\nu \nabla_\mu T^{\mu\nu} &= (\epsilon+\mathcal{A}+P+\Pi)u^\lambda\nabla_\lambda u^\beta + \Delta^{\beta\lambda}\nabla_\lambda(P+\Pi)  - \Delta^\beta_\nu\nabla_\lambda(2\eta\sigma^{\lambda\nu}) \nonumber \\ 
    & + \Delta^\beta_\nu u^\lambda\nabla_\lambda Q^\nu + Q^\beta \nabla_\lambda u^\lambda + Q^\lambda\nabla_\lambda u^\beta = 0, \label{eq:evolution3} 
\end{align}
where Eq. \eqref{Eq:D_T} was projected along and orthogonal to  $u^\mu$. 

On the other hand, following  \cite{Bemfica:2020zjp} one can consider these equations, as 6 independent second-order evolution equation for the 6 independent  variables $(\epsilon, n, u^\mu)$, without assuming the normalization of $u^\mu$; and notice that when $u^\mu$ is normalized, equation \eqref{eq:evolution3} reduce to  only 3 independent equations, since it identically vanish when is  contracted with $u^\mu$. A discussion on the hyperbolicity of this enlarged system is presented in Appendix~\ref{appendix:non_normalized_u_hyperbolicity}.

Equations \eqref{Eq:D_T}, \eqref{Eq:D_J}, or equivalently \eqref{eq:evolution1}–\eqref{eq:evolution3}, can be interpreted as the ideal evolution equations plus small second-order corrections that vanish in the ideal limit:
\begin{align}           
   u^\lambda\nabla_\lambda\epsilon + (\epsilon+P)\nabla_\lambda u^\lambda + \mathcal{O}(\nabla^2) &= 0 \label{eq:IF1} \\
\Delta^{\mu\lambda}\nabla_\lambda P +(\epsilon+P) u^\lambda \nabla_\lambda u^\mu + \mathcal{O}(\nabla^2) &= 0 \label{eq:IF2} \\
    n \nabla_\mu u^\mu + u^\mu\nabla_\mu n + \mathcal{O}(\nabla^2) &= 0 \label{eq:IF3} 
\end{align}
This means that, for some initial data, the solution of these equations correspond to the ideal (zero order) dynamics plus perturbative corrections.
Equivalently, the ideal evolution equations hold modulo $\mathcal{O}(\nabla^2)$ terms, i.e., deviations from the ideal behavior appear only through second-order corrections. This last argument clarifies the motivation for proposing the first-order perturbations \eqref{eq:A}–\eqref{Eq:Q_1} to be proportional to the ideal-fluid equations, since some of these terms can now be regarded as second order in derivatives. This idea will be extremely useful in the calculation of the entropy production.

%%%%%%%%%%%%%%%%%%%%%%%%%%%%%%%%%%%%%%%%%%%%%%%%%%%%%%%%%%%%%%%%%%%%%%%%%%%%%%%%%%%%%%%%%%%%%%%%%%%%%%%%%%%%%%%
\subsection{Hyperbolicity and Causality}
\label{subsection:hyperbolicity_and_causality}
In this section, we follow \cite{Abalos:2026kim}  for the notation and for the definition of strong hyperbolicity, and we derive the conditions under which our system is strongly hyperbolic and causal.

The previous system
of equations can be written as:

\begin{equation}
\mathfrak{N}^{\mu\nu}\left(  U\right)  \partial_{\mu}\partial_{\nu}U+...=0
\label{eq_K_1}%
\end{equation}
where $U\equiv\left(  \epsilon,n,v^{\mu}\right)  $, $\mathfrak{N}^{\mu\nu}$
denote $5\times5$ matrices and the ellipsis indicates terms that do not
contain second derivatives of $U$ (i.e. lower order terms). Here, $v^{\mu}$
represents the three independent components of $u^{\mu}$ after normalization.

Consider a local foliation of the spacetime $M=\cup_{t\in(-T,T)}\Sigma_{t}$\ \ by
spacelike hypersurfaces $\Sigma_{t}$ with unit normal $n_{a}$ orthogonal to
$\Sigma_{t}$, Eq. \eqref{eq_K_1} describes a complete set of evolution equations
whenever $\det\left(  \mathfrak{N}^{\mu\nu}n_{\mu}n_{\nu}\right)  \neq0$. This
condition ensures that the second time derivatives can be isolated from the
remaining derivatives for all dynamical variables. In our case, we take
$n^{\mu}=u^{\mu},$ meaning that time derivatives correspond to derivatives
along the fluid four velocity $u^{\mu}$. We recall that if this condition is
satisfied for $n^{\mu}=u^{\mu}$, it will also hold for others $n^{\mu}$ in a
neighbourhood of $u^{\mu}$. Within our formulation $\det\left(  \mathfrak{N}%
^{\mu\nu}u_{\mu}u_{\nu}\right)  =-\left(  \rho\tau_{Q}\right)  ^{3}%
\tau_{\epsilon}\tau_{n}$ which is nonvanishing provided that $\rho\equiv \epsilon+P,\tau
_{Q},\tau_{\epsilon}$ and $\tau_{n}$ are nonzero; an assumption we adopt
throughout our analysis.

We include here an adapted definition of
strong hyperbolicity with direct application to our formulation. This definition is equivalent to applying the first-order definition to the PDE system obtained via a first-order reduction, in which all derivatives are introduced as independent variables.

\begin{definition}
We say that Eq. \eqref{eq_K_1} are strongly hyperbolic if the following
conditions are satisfied:

1) Real characteristic velocity condition: Defining $l_{\mu}\equiv-\lambda
u_{\mu}+k_{\mu}$  with $k_{\mu}$ orthogonal to $u_{\mu}$ and normalized
$\left\vert k\right\vert =1$ (these two assumptions on $k_{\mu}$ hold
throughout the paper), the characteristic velocities, equivalently, the roots
of the polynomial $p\left(  \lambda\right)  \equiv\det\left(  \mathfrak{N}%
^{\mu\nu}l_{\mu}l_{\nu}\right)  $ are real for all $k_{\mu}$.

2) Multiplicity condition: For all $k_{\mu}$, if $\lambda_{i}$ is a root of
$p\left(  \lambda\right)  $ with degeneracy $q_{i}$ then $q_{i}=\dim\left(
\ker\left(  \mathfrak{N}^{\mu\nu}\left.  l_{\mu}\right\vert _{\lambda_{i}%
}\left.  l_{\nu}\right\vert _{\lambda_{i}}\right)  \right)  $.

3) Uniformity condition: The solutions $\delta U$ of the characteristic
equation $\mathfrak{N}^{\mu\nu}\left(  U_{0}\right)  l_{\mu}l_{\nu}\delta
U=0$, where $U_{0}$ denotes a smooth background solution, depend continuously
on $k_{\mu}$ for all $k_{\mu}$.
\end{definition}

We notice that $S\left(  \lambda\right)  \equiv\mathfrak{N}^{\mu\nu}l_{\mu
}l_{\nu}$ is called the principal symbol, or equivalently, the second order matrix
pencil of the system, since it is a matrix quadratic in $\lambda$.

The characteristic velocities, given by the roots of the polynomial $p\left(
\lambda\right)  $, represent the propagation speeds of high frequency
perturbations. To ensure a physically meaningful theory, one must impose the
causality condition, which requires all characteristic velocities to remain
below the speed of light; in our units, smaller than one.

We now introduce some definitions that will be used throughout the paper:%
\begin{align}
A  &  \equiv-\rho\tau_{\epsilon}\tau_{n}\tau_{Q}\\
B  &  \equiv V\tau_{\epsilon}\tau_{n}+(n\beta_{n}+\rho\beta_{\epsilon}%
)\tau_{\epsilon}\tau_{n}+\rho\tau_{Q}(\tau_{n}\tau_{P}-\lambda_{n}%
\tau_{\epsilon})+n\beta_{n}\tau_{\mathcal{J}}\tau_{\epsilon}\\
C  &  \equiv V(\beta_{\epsilon}\tau_{n}+\lambda_{n}\tau_{\epsilon}%
)-(n\beta_{n}+\rho\beta_{\epsilon})\tau_{n}\tau_{P}+\rho\lambda_{n}\tau
_{P}\tau_{Q}+\rho(\beta_{\epsilon}\lambda_{n}-\beta_{n}\lambda_{\epsilon}%
)\tau_{\epsilon}-n\beta_{n}\tau_{\mathcal{J}}\tau_{P}\\
D  &  \equiv(\rho\tau_{P}-V)(\beta_{n}\lambda_{\epsilon}-\beta_{\epsilon
}\lambda_{n}) \label{eq:D}
\end{align}

where $V\equiv\zeta+\frac{4}{3}\eta$ and%

\[
\Delta\equiv18ABCD+B^{2}C^{2}-27A^{2}D^{2}-4AC^{3}-4B^{3}D.
\]

For simplicity, we assume that 
\[
0<\tau_{\epsilon}%
,\tau_{n},\tau_{Q},\tau_\mathcal{J},\tau_P \ \ \text{which implies} \ \ A<0.
\]

In the following lemma, we state sufficient conditions for ensuring strong
hyperbolicity and causality of our BDNK formulation.

\begin{lemma} \label{Lemma:hip:cau}
If the following inequalities are satisfied,
\begin{align}
0  &  <\frac{\eta}{\rho\tau_{Q}}\label{eq_con_1}\\
0  &  <\Delta\label{eq_con_2}\\
0  &  <B\label{eq_con_3}\\
C  &  <0\label{eq_con_4}\\
0  &  <D \label{eq_con_5}%
\end{align}
then Eqs. \eqref{eq:evolution1}-\eqref{eq:evolution3}, with dissipative terms expressed as Eqs. \eqref{eq:A}-\eqref{eq:J}, are
strongly hyperbolic. Moreover, if the additional conditions
\begin{align}
\frac{\eta}{\rho\tau_{Q}}  &  <1\label{eq_con_6}\\
3A+2B+C  &  <0\label{eq_con_7}\\
3A+B  &  <0\label{eq_con_8}\\
A+B+C+D  &  <0 \label{eq_con_9}%
\end{align}
are also satisfied, then the system is causal.
\end{lemma}

The conditions $0\leq\Delta,$ $0\leq B,$ $C\leq0$ and $0\leq D$ guarantee real
characteristic velocities. However, we impose the strict inequalities
$0<\Delta,$ $0<B,$ $C<0$, $0<D$ since they reduce algebraic degeneracy,
simplifying the analysis of the multiplicity and uniformity conditions when the
definition of strong hyperbolicity is checked. Consequently, to obtain necessary and
sufficient conditions, one should analyze the cases in which $\Delta$, $B$,
$C$ and $D$ may vanish.

\bigskip

\begin{proof}
We begin by showing that the real characteristic velocity condition in the
definition of strong hyperbolicity is satisfied under the assumption
\eqref{eq_con_1}-\eqref{eq_con_5}%
. Next, we show that \eqref{eq_con_6}-\eqref{eq_con_9} ensure the causality condition. Finally, we conclude the proof
of strong hyperbolicity by showing that \eqref{eq_con_1}-\eqref{eq_con_5} also guarantee the multiplicity
and uniformity conditions.

The principal part of the fluid evolution equations is%

\[
\mathfrak{N}^{\alpha\beta}l_{\alpha}l_{\beta}=\left[
\begin{array}
[c]{ccc}%
\lambda_{\epsilon}\Delta^{\alpha\beta} & \tau_{n}u^{\alpha}u^{\beta}%
+\lambda_{n}\Delta^{\alpha\beta} & n(\tau_{n}+\tau_{\mathcal{J}})\delta_{\nu
}^{(\alpha}u^{\beta)}\\
\tau_{\epsilon}u^{\alpha}u^{\beta}+\beta_{\epsilon}\Delta^{\alpha\beta} &
\beta_{n}\Delta^{\alpha\beta} & \rho(\tau_{\epsilon}+\tau_{Q})u^{(\alpha
}\delta_{\nu}^{\beta)}\\
(\tau_{P}+\beta_{\epsilon})u^{(\alpha}\Delta^{\beta)\mu} & \beta_{n}%
u^{(\alpha}\Delta^{\beta)\mu} & C_{\nu}^{\mu\alpha\beta}%
\end{array}
\right]  l_{\alpha}l_{\beta}%
\]
with $l_{\mu}=-\lambda u_{\mu}+k_{\mu}$ \ \ (here $u_\mu k^\mu=0$ and $\left\vert
k\right\vert ^{2}=1$), $\rho \equiv\epsilon+P$, and%
\begin{equation}
C_{\nu}^{\mu\alpha\beta}=\left[  \rho\tau_{P}-\zeta-\frac{1}{3}\eta\right]
\Delta^{\mu(\alpha}\delta_{\nu}^{\beta)}+[\rho\tau_{Q}u^{\alpha}u^{\beta}%
-\eta\Delta^{\alpha\beta}]\delta_{\nu}^{\mu}.
\end{equation}
Taking the determinant we obtain the tenth order polynomial $p\left(
\lambda\right)  $, this is
\begin{equation}
p\left(  \lambda\right)  \equiv\det\left(  \mathfrak{N}^{\mu\nu}l_{\mu}l_{\nu
}\right)  =\left(  \rho\tau_{Q}\right)  ^{2}\left(  g_{shear}^{\mu\nu}l_{\mu
}l_{\nu}\right)  ^{2}f\left(  \lambda\right)  \label{eq_det_1}%
\end{equation}
where the shear "effective" metric is defined as%
\[
g_{shear}^{\mu\nu}\equiv c_{sh}^{2}\Delta^{\mu\nu}-u^{\mu}u^{\nu}%
\]
with
\[
c_{sh}\equiv\sqrt{\frac{\eta}{\rho\tau_{Q}}}.
\]
We notice that from condition \ref{eq_con_1} guarantees $c_{sh}$ real. The remaining
part of the determinant is
\[
f\left(  \lambda\right)  \equiv\rho\tau_{Q}\tau_{\epsilon}\tau_{n}\left(
g_{sound_{1}}^{\mu\nu}l_{\mu}l_{\nu}\right)  \left(  g_{sound_{2}}^{\mu\nu
}l_{\mu}l_{\nu}\right)  \left(  g_{sound_{3}}^{\mu\nu}l_{\mu}l_{\nu}\right)
\]
where the sound ``effective'' metrics are given by
\[
g_{sound_{i}}^{\mu\nu}\equiv c_{s_{i}}^{2}\Delta^{\mu\nu}-u^{\mu}u^{\nu}%
\]
with $i=1,2,3$, and where $c_{s_{1}}^{2},c_{s_{2}}^{2},c_{s_{3}}^{2}$ are the
roots of the cubic polynomial
\begin{equation}
f\left(  x\right)  \equiv Ax^{3}+Bx^{2}+Cx+D, \label{eq_f_1}% 
\end{equation}
with $x=\lambda^{2}$. Using $g_{shear}^{\mu\nu}l_{\mu}l_{\nu}=-\lambda
^{2}+c_{sh}^{2}$, $g_{sound_{i}}^{\mu\nu}l_{\mu}l_{\nu}=-\lambda^{2}+c_{s_{i}%
}^{2}$, and combining these results with Eq. \eqref{eq_det_1}, we conclude that
the ten root of $p\left(  \lambda\right)  $ are $\pm c_{sh},\pm c_{sh},\pm
c_{s_{1}},\pm c_{s_{2}},\pm c_{s_{3}}$. Geometrically, this means that
$\det\left(  \mathfrak{N}^{\mu\nu}l_{\mu}l_{\nu}\right)  =0$ whenever $l_{\mu
}$ is null with respect to the shear effective metric, or to one of the sound
effective metrics.

In order to satisfy the real characteristic velocity condition, we must show
that $f\left(  x\right)  $ has real and positive roots. From \cite{dickson1922first}, we know
that if the discriminant $\Delta$ is positive, as assumed in \ref{eq_con_2},
then all the roots of $f\left(  x\right)  $ are real and distinct. To verify
that they are positive, we note that, under the assumptions \ref{eq_con_3}-\ref{eq_con_5}, one finds $f\left(  x\right)  >0$ for
$x\in(-\infty,0]$ since each term in Eq. \eqref{eq_f_1} is positive. This
implies that the three roots \ $c_{s_{1}}^{2},c_{s_{2}}^{2},c_{s_{3}}^{2}$ of
$f\left(  x\right)  $\ lie in the interval $(0,\infty)$, thereby satisfying
the real characteristic velocity condition.

The explicit expressions for the characteristic velocities of the sound
channel $c_{s_{i}}$ are given by%
\begin{equation}
c_{s_{i}}=\sqrt{R\cos\left(  \frac{\theta+2\pi\left(  k-1\right)  }{3}\right)
-\frac{B}{3A}},\quad k=1,2,3 \label{eq:sound_velocities}
\end{equation}
with%
\begin{equation}
R=2\sqrt{-\frac{p}{3}},\quad\theta=\arccos\left(  \frac{3q}{2p}\sqrt{-\frac
{3}{p}}{}\right),
\end{equation}
and%
\begin{equation}
p=\frac{3AC-B^{2}}{3A^{2}},\quad q=\frac{2B^{3}-9ABC+27A^{2}D}{27A^{3}}.
\end{equation}

Now we analyze the causality condition. The inequality \ref{eq_con_6} ensures
that $c_{sh}<1$; it only remains to show that $c_{s_{i}}<1$ under the
assumptions \ref{eq_con_7}-\ref{eq_con_9}.

We consider the translation $x=y+1$ in the function $f\left(  x\right)  $,
which gives%
\[
f\left(  y\right)  =Ay^{3}+\left(  3A+B\right)  y^{2}+\left(  3A+2B+C\right)
y+\left(  A+B+C+D\right)  .
\]
We know that the roots of this polynomial lie in the interval $\left(
-1,\infty\right)  $. However, assuming \ref{eq_con_7}-\ref{eq_con_9} and proceeding as before, it is straightforward to verify that $f\left(  y\right)  <0$ 
for all $y\in\lbrack0,\infty)$. We therefore conclude that the roots of $f\left(
y\right)  $ lie in $\left(  -1,0\right)  $, which implies that the roots of
$f\left(  x\right)  $ lie in $\left(  0,1\right)  $, thus completing the proof
of causality.

To verify the multiplicity and the uniformity conditions, we examine  the
characteristic equation
\[
\mathfrak{N}^{\mu\nu}\left(  U_{0}\right)  l_{\mu}l_{\nu}\delta U=0
\]
where $\delta U=\left(  \delta\epsilon,\delta n,\delta v^{\mu}\right)  $ and
$U_{0}$ denotes a smooth background solution. We consider the orthonormal
basis $\left\{  u^{\mu},k^{\mu},e_{1}^{\mu},e_{2}^{\mu}\right\}  $ and
decompose the vector $\delta v^{\mu}$ as $\delta v^{\mu}=k^{\mu}\delta
a+e_{1}^{\mu}\delta b+e_{2}^{\mu}\delta c$ (note that there is no component along $u^\mu$, since $u^\mu$ is normalized). In this basis, the principal
symbol can be written as
\begin{equation}
\mathfrak{N}^{\mu\nu}l_{\mu}l_{\nu}=\left[
\begin{array}
[c]{ccc|cc}%
\lambda_{\epsilon} & \tau_{n}\lambda^{2}+\lambda_{n} & n(\tau_{n}%
+\tau_{\mathcal{J}})\lambda & 0 & 0\\
\tau_{\epsilon}\lambda^{2}+\beta_{\epsilon} & \beta_{n} & \rho(\tau_{\epsilon
}+\tau_{Q})\lambda & 0 & 0\\
(\tau_{P}+\beta_{\epsilon})\lambda & \beta_{n}\lambda & \left(  \rho\tau_{P}%
-\zeta-\frac{4}{3}\eta\right)  +\rho\tau_{Q}\lambda^{2} & 0 & 0\\\hline
0 & 0 & 0 & \rho\tau_{Q}\lambda^{2}-\eta & 0\\
0 & 0 & 0 & 0 & \rho\tau_{Q}\lambda^{2}-\eta
\end{array}
\right]  . \label{eq:principal_symbol}
\end{equation}

We notice that this matrix splits into two blocks, the sound and shear blocks,
which are decoupled. It is then straightforward to conclude that, for
$\lambda=c_{sh}$, the kernel of the matrix is given by%
\begin{equation}
\delta U^{(\text{shear})}=\left[
\begin{array}
[c]{c}%
0\\
0\\
e_{1}^{\nu}%
\end{array}
\right]  ,\left[
\begin{array}
[c]{c}%
0\\
0\\
e_{2}^{\nu}%
\end{array}
\right]  \label{vec_1}%
\end{equation}
showing that the multiplicity condition is satisfied for $\lambda=c_{sh}$.
These vectors  are continuous with respect to $k^{\mu}$ and the same vectors
are obtained as the kernel of the principal symbol when $\lambda=-c_{sh}$.

We now  consider the sound block for  $\lambda=c_{s_{i}}$, with $i=1,2,3$. In
this case, for each $c_{s_{i}}$, the associated determinant vanishes, which implies the existence of
at least one vector in the kernel,
\begin{equation}
\delta U^{(\text{sound})}=\left[
\begin{array}
[c]{c}%
\delta\epsilon\\
\delta n\\
k^{\mu}\delta a
\end{array}
\right]  ,\label{vec_2}%
\end{equation}
for some value of $\delta\epsilon$, $\delta n$ and $\delta a$. Moreover, this
vector is continuous with respect to $k^{\mu}$. Finally, since we are assuming
$0<\Delta$, the degeneracy of each $c_{s_{i}}$ is $1$ (or $3$ if $c_{s_{i}%
}=c_{sh}$). Therefore, since  \eqref{vec_2} (and \ref{vec_1} when $c_{s_{i}%
}=c_{sh}$) exists, the multiplicity condition is satisfied for each $c_{s_{i}%
}$. The same conclusion holds for $\lambda=-c_{s_{i}}$. We thus conclude that
the multiplicity condition is satisfied for all characteristic velocities.

The final step in establishing strong hyperbolicity is the uniformity
condition. As mentioned above, the kernels $\delta U^{(\text{shear})}$ and
$\delta U^{(\text{sound})}$ are continuous with respect to $k^{\mu}$ and for
all $k^{\mu}$. Therefore the systems is strongly hyperbolic.
\end{proof}

From a geometrical point of view, the sufficient conditions \ref{eq_con_1}%
-\ref{eq_con_5} for strong hyperbolicity are equivalent to requiring that the
effective metrics $g_{shear}^{\mu\nu}$ and $g_{sound_{i}}^{\mu\nu}$ are real
and lorentzian, with $u_{\mu}$ timelike in all them, and with mutually
non-intersecting null cones of $g_{sound_{i}}^{\mu\nu}$ for different $i$.

As a final remark, we notice that the expression for $g_{sound_{i}}^{\mu\nu}$ 
is exactly the same as the one obtained in the analysis of hyperbolicity for an
ideal fluid, upon replacing $c_{s_{i}}$ with the adiabatic sound speed $c_{s}\equiv\left(\frac{\partial P}{\partial \epsilon}\right)_{n}
+\frac{n}{\rho}\left(\frac{\partial P}{\partial n}\right)_{\epsilon}$. We expect then $c_{s_{i}}$ to represent some perturbations of $c_{s}$.

We conclude this subsection with a  discussion of the case
$D=0$. In this situation the eigenvalues become degenerate, with the appearance of a characteristic mode $c_{s_{1}}=0$ with degeneracy two. This scenario does not break the hyperbolicity and causality of the system, however, it requires a separate analysis to stablish strong hyperbolicity. As we mention,  $\lambda=c_{s_{1}}=0$ has degeneracy two, then one must verify that the corresponding eigenspace ($\mathfrak{N}^{\mu\nu}  l_{\mu}l_{\nu} \delta U=0$) is also two-dimensional in order to retain the strong hyperbolicity.

\begin{lemma} 
Asumming $D=0$ and conditions \ref{eq_con_1}-\ref{eq_con_4}, 
    Eqs. \eqref{eq:evolution1}-\eqref{eq:evolution3} are strongly hyperbolic if and only if $\rho \tau_P = V$ and $\beta_n\lambda_\epsilon - \beta_\epsilon\lambda_n=0$. \label{lemma_Dequa0}
\end{lemma}

\begin{proof}
    When $D=0$,  the characteristic polynomial Eq. \eqref{eq_f_1}, gives $\lambda^2=0$ as solutions. The associated principal symbol of this characteristic mode is

\begin{equation}
    \mathfrak{N}^{\mu\nu}l_{\mu}l_{\nu}=\left[
\begin{array}
[c]{ccc|cc}%
\lambda_{\epsilon} & \lambda_{n} & 0 & 0 & 0\\
\beta_{\epsilon} & \beta_{n} & 0 & 0 & 0\\
0 & 0 &   \rho\tau_{P}%
-V & 0 & 0\\\hline
0 & 0 & 0 & -\eta & 0\\
0 & 0 & 0 & 0 & -\eta
\end{array}
\right] \label{eq:principal_symbol_degeneracy2}
\end{equation}

In order to have strong hyperbolicity we need the kernel of this matrix to have dimension two. The shear sector of the system (lower right block) is clearly non degenerate since $\eta>0$ from the hyperbolicity conditions. We then focus on the sound sector (upper left block). There is only one way for this block to have a two-dimensional kernel: $\beta_n\lambda_\epsilon-\beta_\epsilon\lambda_n=0$ and $\rho\tau_P=V$, in agreement with the statement of the Lemma.
In this case the kernel vectors are

\begin{equation}
\delta U^{(\text{sound})}_{\lambda=0}=\left[
\begin{array}
[c]{c}%
0\\
0\\
k^{\nu}%
\end{array}
\right]  ,\left[
\begin{array}
[c]{c}%
-\beta_n \\
\beta_\epsilon\\
0%
\end{array}
\right]  \label{vec_1_b}%
\end{equation}
which depend continuously on $k^\nu$, satisfying the uniformity conditions and therefore all the conditions for strong hyperbolicity.

We notice that the case when $\beta_ \epsilon=\beta_n=\lambda
_\epsilon=\lambda_n=0$ should be excluded, since it yields $C=0$, increasing the degeneracy to 4 ($\lambda^4=0$), making the system weakly hyperbolic.

\end{proof}

To understand how the case $D=0$ translates in terms of diffusion parameters, we have to look at the expression of D, plugging the expressions of $\lambda_{n,\epsilon}$ and $\beta_{n,\epsilon}$ from Eqs. \eqref{lambda_1}-\eqref{beta_2} into Eq. \eqref{eq:D}, we obtain:
\begin{equation}
 D = (\rho \tau_P-V) (\sigma \tau_\mathcal{J}+\sigma_0 \tau_Q) s_2,  
\label{eq:D_expanded}
 \end{equation}
where
\begin{equation}
    s_2 \equiv T\left( \left. \frac{\partial (\mu/T)}{\partial n} \right|_\epsilon \left. \frac{\partial P}{\partial \epsilon} \right|_n -  \left. \frac{\partial (\mu/T)}{\partial \epsilon} \right|_n \left. \frac{\partial P}{\partial n} \right|_\epsilon \right). \label{in_s_2}
\end{equation}
 
As explained in Appendix~\ref{appendix:positivity_of_detJ(e,n)}, the strict concavity of the entropy density is equivalent to  the positivity of $s_2$. This strict concavity condition is a standard local thermodynamic stability requirement, it ensures that equilibrium is a local entropy maximum \cite{callen1993thermodynamics, lukavcova2025note}. We therefore impose this condition as restrictions for the equation of state. So, assuming $s_2>0$, condition  
$\beta_n\lambda_\epsilon - \beta_\epsilon\lambda_n=0$ translates into $\sigma \tau_\mathcal{J}+\sigma_0 \tau_Q=0$.
At this point, if one wants to keep the freedom of having $\rho \tau_P\neq V$, the condition $\sigma \tau_\mathcal{J}\neq-\sigma_0\tau_Q$ must be imposed. Alternatively, one can impose $\rho\tau_P=V$ and $\sigma\tau_\mathcal{J}=-\sigma_0\tau_Q$ and preserve the strong hyperbolicty even with $D=0$, as shown in Lemma \ref{lemma_Dequa0}.

In Appendix \ref{appendix:non_normalized_u_hyperbolicity}, we show that the results obtained in this section also hold without imposing the normalization of $u^\mu$.

%%%%%%%%%%%%%%%%%%%%%%%%%%%%%%%%%%%%%%%%%%%%%%%%%%%%%%%%%%%%%%%%%%%%%%%%%%%%%%%%%%%%%%%%%%%%%%%%%%%%%%%%%%%
\subsection{Entropy production}

For our theory to be physically consistent, we must show that the second principle of thermodynamics is satisfied, i.e. the local entropy production is non-negative. This is equivalent to demanding $\nabla_{\mu} S^{\mu} \geq 0$, where $S^{\mu}$ denotes the entropy current. As we work in a gradient expansion scheme, this requirement will be imposed up to $\mathcal{O}(\nabla^{3})$.

We begin by stating a lemma that establishes an identity for the divergence of the entropy current, derived directly from the first law of thermodynamics, without any approximations.

\begin{lemma}
Considering the general definitions of the stress-energy tensor $T^{\mu\nu}$ and the charge current $J^{\mu}$ (Eqs.~\eqref{eq:T_1} and~\eqref{eq:J_1}), without imposing any specific prescription for the first-order corrections, the entropy current
\begin{equation}
TS^{\mu}\equiv Pu^{\mu}-u_{\nu}T^{\mu\nu}-\mu J^{\mu}%
.\label{eq:entropy current}%
\end{equation}
satisfies the following exact identity:
\begin{align}
T\nabla_{\mu}S^{\mu} &  =-\mathcal{A}\,\left(  \frac{u^{\mu}\nabla_{\mu}T}%
{T}\right)  -\Pi\,\left(  \nabla.u\right)  -\mathcal{T}^{\mu\nu}\,\sigma
_{\mu\nu}\nonumber\\
&  -\left[  \frac{\nabla^{\mu}T}{T}+u^{\nu}\nabla_{\nu}u^{\mu}\right]  Q_{\mu
}-T\left(  N\,u^{\mu}+\mathcal{J}^{\mu}\right)  \nabla_{\mu}\left(  \frac{\mu
}{T}\right)  .\label{Eq:div_entropy_lemma}%
\end{align}

\end{lemma}

In the ideal fluid case, the expression of entropy current in \eqref{eq:entropy current} gives the well known thermodynamical entropy density at equilibrium $s$, as the component along $u^\mu$, see Eq. \eqref{eq:id_1}. Constituting the most straightforward covariant generalization of entropy current in terms of $T^{\mu\nu}$ and $J^\mu$ (see \cite{Bhattacharya:2011tra} for more details).

The proofs in this section follow an approach similar to that in Refs.~\cite{Bemfica:2020zjp,Kovtun:2019hdm}, generalized to include the contribution from the charge current.

\begin{proof}
We split the energy--momentum tensor and the charge current into their ideal and dissipative parts,
\[
T^{\mu\nu} = T_{(0)}^{\mu\nu} + T_{(1)}^{\mu\nu}, \qquad 
J^{\mu} = J_{(0)}^{\mu} + J_{(1)}^{\mu},
\]
with $T_{(0)}^{\mu\nu} = \epsilon\, u^{\mu}u^{\nu} + P\,\Delta^{\mu\nu}$ and
$J_{(0)}^{\mu} = n\,u^{\mu}$. 
Since $u_{\nu}T^{\mu\nu} = -\epsilon\,u^{\mu} + u_{\nu}T_{(1)}^{\mu\nu}$ and using the thermodynamic identity for the canonical entropy density at equilibrium $s$
\begin{equation}
Ts = \epsilon + P - \mu n, 
\label{eq:id_1}
\end{equation}
(which follows from the first law of thermodynamics and the Euler relation for extensivity~\cite{callen1993thermodynamics}), Eq.~\eqref{eq:entropy current} can be written as
\[
T S^{\mu} = T s\,u^{\mu} - u_{\nu}T_{(1)}^{\mu\nu} - \mu J_{(1)}^{\mu}.
\]
Dividing by $T$ and taking the divergence, we obtain
\begin{equation}
\nabla_{\mu} S^{\mu} =
\left[ 
\nabla_{\mu}(s u^{\mu}) 
- \frac{u_{\nu}}{T}\nabla_{\mu}T_{(1)}^{\mu\nu}
- \frac{\mu}{T}\nabla_{\mu}J_{(1)}^{\mu}
\right]
- T_{(1)}^{\mu\nu}\nabla_{\mu}\!\left(\frac{u_{\nu}}{T}\right)
- J_{(1)}^{\mu}\nabla_{\mu}\!\left(\frac{\mu}{T}\right),
\label{Eq:div_S}
\end{equation}
where we will show that the term in brackets vanishes.

Using the first law of thermodynamics, $d\epsilon = T\,ds + \mu\,dn$
(again, obtained from the first principle of thermodynamics~\cite{callen1993thermodynamics}),
and projecting along $u^{\mu}$, we find
\[
T\,u^{\mu}\nabla_{\mu}s 
= u^{\mu}\nabla_{\mu}\epsilon - \mu\,u^{\mu}\nabla_{\mu}n,
\]
which gives
\begin{equation}
T\,\nabla_{\mu}(s u^{\mu}) 
= u^{\mu}\nabla_{\mu}\epsilon 
- \mu\,u^{\mu}\nabla_{\mu}n 
+ (\epsilon + P - \mu n)\,(\nabla\!\cdot\!u),
\label{Eq:div_1}
\end{equation}
where Eq.~\eqref{eq:id_1} has been used.

From the energy--momentum and charge current conservation,
$\nabla_{\mu}T^{\mu\nu}=0$ and $\nabla_{\mu}J^{\mu}=0$, we obtain
\begin{align*}
u^{\mu}\nabla_{\mu}\epsilon + (\epsilon+P)(\nabla\!\cdot\!u)
&= u_{\nu}\nabla_{\mu}T_{(1)}^{\mu\nu}, \\
u^{\mu}\nabla_{\mu}n + n(\nabla\!\cdot\!u)
&= -\,\nabla_{\mu}J_{(1)}^{\mu},
\end{align*}
where the first relation follows from contracting $u_{\nu}$ with
$\nabla_{\mu}T^{\mu\nu}=0$.

Substituting these two relations into Eq.~\eqref{Eq:div_1}, we find
\[
T\,\nabla_{\mu}(s u^{\mu})
= u_{\nu}\nabla_{\mu}T_{(1)}^{\mu\nu}
+ \mu\,\nabla_{\mu}J_{(1)}^{\mu}.
\]
This last expression shows that the term inside the brackets in
Eq.~\eqref{Eq:div_S} vanishes, leading to
\begin{align}
\nabla_{\mu}S^{\mu}
&= -\,T_{(1)}^{\mu\nu}\nabla_{\mu}\!\left(\frac{u_{\nu}}{T}\right)
- J_{(1)}^{\mu}\nabla_{\mu}\!\left(\frac{\mu}{T}\right),
\label{Eq:div_S_1}\\
&= \frac{1}{T^{2}}\,u_{\nu}T_{(1)}^{\mu\nu}\nabla_{\mu}T
- \frac{1}{T}\,T_{(1)}^{\mu\nu}(\nabla_{\mu}u_{\nu})
- J_{(1)}^{\mu}\nabla_{\mu}\!\left(\frac{\mu}{T}\right),
\label{Eq:div_S_2}
\end{align}
where in the last line we have expanded the derivative
$\nabla_{\mu}(u_{\nu}/T)$.

Now, we use the expressions 
\begin{align*}
T_{(1)}^{\mu\nu} &  =\mathcal{A}u^{\mu}u^{\nu}+\Pi\Delta^{\mu\nu}+u^{\mu
}Q^{\nu}+u^{\nu}Q^{\mu}+\mathcal{T}^{\mu\nu},\\
J_{(1)}^{\mu} &  =Nu^{\mu}+\mathcal{J}^{\mu},
\end{align*}
from where we obtain 
\begin{align}
u_{\nu}T_{(1)}^{\mu\nu}\nabla_{\mu}T &  =-\mathcal{A}u^{\mu}\nabla_{\mu
}T-Q^{\mu}\nabla_{\mu}T, \\
T_{(1)}^{\mu\nu}\left(  \nabla_{\mu}u_{\nu}\right)   &  =\Pi\left(
\nabla.u\right)  +Q_{\nu} u^\mu \nabla_\mu u ^\nu +\mathcal{T}^{\mu\nu}\sigma_{\mu\nu
}, \label{Eq:T_1_d_u}\\
J_{(1)}^{\mu}\nabla_{\mu}\!\left(  \frac{\mu}{T}\right)   &  =\left(  Nu^{\mu
}+\mathcal{J}^{\mu}\right)  \nabla_{\mu}\left(  \frac{\mu}{T}\right).
\end{align}
These relations lead to Eq.~\eqref{Eq:div_entropy_lemma}.

To show Eq.~\eqref{Eq:T_1_d_u}, it is useful to decompose the velocity gradient
into its irreducible parts as
\[
\nabla_{\mu}u_{\nu}
= \sigma_{\mu\nu} + \omega_{\mu\nu}
+ \tfrac{1}{3}(\nabla\!\cdot\!u)\,\Delta_{\mu\nu}
- u_{\mu}(u^{\lambda}\nabla_{\lambda}u_{\nu}),
\]
where $\sigma_{\mu\nu}$ is the shear tensor defined in
Eq.~\eqref{Eq:sigma_ab} (symmetric, traceless, and orthogonal to $u^{\mu}$),
and $\omega_{\mu\nu}\equiv \Delta_{[\mu}^{\alpha}\Delta_{\nu]}^{\beta}
\nabla_{\alpha}u_{\beta}$ is the vorticity tensor (antisymmetric and orthogonal to $u^{\mu}$)

\end{proof}

We apply now this result to our system.

\begin{lemma} \label{Lemma:entropy}
Given the evolution equations \eqref{eq:evolution1}--\eqref{eq:evolution3},
with the dissipative terms expressed as in \eqref{eq:A}--\eqref{eq:J},
if the following conditions hold:
\begin{equation}
\zeta \geq 0, \qquad \eta \geq 0, \qquad \sigma + \sigma_0\geq 0,
\label{eq:entropy_conditions}
\end{equation}
then the entropy current $S^\mu$, defined in Eq.~\eqref{eq:entropy current},
satisfies
\[
0 \leq \nabla_\mu S^\mu + \mathcal{O}(\nabla^3).
\]
In other words, the entropy production is non--negative up to second order
in the gradient expansion.
\end{lemma}

\begin{proof}
Assuming that, for some initial data, the second--order gradient corrections
only induce small deviations from the ideal fluid, we can use
Eqs.~\eqref{eq:evolution1}--\eqref{eq:evolution3} together with the
definitions~\eqref{eq:A}--\eqref{Eq:Q_1} to obtain the relations
\begin{align}
\mathcal{A} &  =0+\mathcal{O}\left(  \nabla^{2}\right)  \label{Eq:A_2}\\
\Pi &  =-\xi\left(  \nabla.u\right)  +\mathcal{O}\left(  \nabla^{2}\right) \\
N &  =0+\mathcal{O}\left(  \nabla^{2}\right)  \\
\mathcal{J}^{\mu} &  = -\sigma_{0}T\Delta^{\mu\lambda}\nabla_{\lambda}\left(
\frac{\mu}{T}\right) + \mathcal{O}(\nabla^2)  \\
Q^{\mu} &  =\sigma T\frac{\left(  \epsilon+P\right)  }{n}\Delta^{\mu\lambda
}\nabla_{\lambda}\left(  \frac{\mu}{T}\right)  +\mathcal{O}\left(  \nabla
^{2}\right) \\
\mathcal{T}^{\mu\nu} &  =-2\eta\sigma^{\mu\nu}+\mathcal{O}\left(  \nabla
^{2}\right)  \label{Eq:T_2}%
\end{align}
In addition, we have the auxiliary relation
\begin{equation}
-\,T\,\frac{n}{\epsilon + P}\,\Delta^{\mu\lambda}\nabla_{\lambda}
\!\left(\frac{\mu}{T}\right)
= \Delta^{\mu\lambda}\frac{\nabla_{\lambda}T}{T}
+ u^{\nu}\nabla_{\nu}u^{\mu}
+ \mathcal{O}(\nabla^2),
\label{Eq:ident_1}
\end{equation}
which is obtained by rewriting the Gibbs--Duhem relation
$dP = n\,d\mu + s\,dT$ (following from the first law of thermodynamics \cite{callen1993thermodynamics}) as
\[
-\,T\,\frac{n}{\epsilon + P}\,d\!\left(\frac{\mu}{T}\right)
= \frac{dT}{T} - \frac{dP}{\epsilon + P},
\]
or equivalently, in its hypersurface projected form
\[
-\,T\,\frac{n}{\epsilon + P}\,\Delta^{\mu\lambda}\nabla_{\lambda}
\!\left(\frac{\mu}{T}\right)
= \Delta^{\mu\lambda}\frac{\nabla_{\lambda}T}{T}
- \frac{\Delta^{\mu\lambda}\nabla_{\lambda}P}{\epsilon + P}.
\]
Then equation~\eqref{Eq:ident_1} follows by using Eq.~\eqref{eq:IF2}
to replace the pressure gradient term.

%%%%%%%%%%%%%%%%%%

Finally, substituting Eqs.~\eqref{Eq:A_2}--\eqref{Eq:T_2} into
Eq.~\eqref{Eq:div_entropy_lemma} yields
\begin{align}
T\,\nabla_{\mu}S^{\mu}
&= \zeta\,(\nabla\!\cdot\!u)^2
+ 2\eta\,\sigma^{\mu\nu}\sigma_{\mu\nu}
- \Bigl[\frac{\nabla_{\mu}T}{T}
+ u^{\nu}\nabla_{\nu}u_{\mu}\Bigr]
\,\sigma\,T\,\frac{\epsilon + P}{n}\,
\Delta^{\mu\lambda}\nabla_{\lambda}\!\left(\frac{\mu}{T}\right)
\nonumber\\[4pt]
&\quad + T^{2}\sigma_{0}\,
\Bigl[\nabla_{\mu}\!\left(\frac{\mu}{T}\right)
\Delta^{\mu\lambda}\nabla_{\lambda}\!\left(\frac{\mu}{T}\right)\Bigr]
+ \mathcal{O}(\nabla^{3}) \nonumber\\[4pt]
&= \zeta\,(\nabla\!\cdot\!u)^2
+ 2\eta\,\sigma^{\mu\nu}\sigma_{\mu\nu}
+ T^{2}(\sigma + \sigma_{0})\,
\Bigl[\nabla_{\mu}\!\left(\frac{\mu}{T}\right)
\Delta^{\mu\lambda}\nabla_{\lambda}\!\left(\frac{\mu}{T}\right)\Bigr]
+ \mathcal{O}(\nabla^{3}),
\end{align}
where in the last equality Eq.~\eqref{Eq:ident_1} has been used.
The positivity of this expression up to second order follows immediately
from the inequalities~\eqref{eq:entropy_conditions}, proving that the entropy
production is non--negative at this order.
\end{proof}

In particular, $\sigma + \sigma_0 = \sigma_E\geq 0$ (the last inequality of \eqref{eq:entropy_conditions}) states that the amount of energy transported by heat (excluding the one brought by particles) should follow an Eckart law up to $\mathcal{O}(\nabla^2)$, with positive heat diffusion coefficient. This is equivalent to the result found in \cite{Bemfica:2020zjp}.

\subsection{Linear Stability}

The last fundamental property we want our model to have is the linear stability around thermal equilibrium uniform solutions. In other words, we want to find conditions ensuring that small perturbations of homogeneous solutions decay in time, so that the system returns to its original configuration. This is not guaranteed by strong hyperbolicity (well-posedness), since a strongly hyperbolic model may still admit perturbations  that grow exponentially. The investigation of linear stability requires indeed to go beyond the principal part, taking also into account for the lower order damping terms. 

Let $\mathbf{q}=(\epsilon,n,u^\mu)$ be an homogeneous equilibrium configuration with fluid at rest $u^\mu=(1,0,0,0)$ on a Minkowski spacetime $g_{\mu\nu}=\textrm{diag}(-1,1,1,1)$, and let us apply a perturbation $\delta \mathbf{q}=(\delta \epsilon, \delta n, \delta u^\mu)$, with $u_\mu \delta u^\mu =0$ (so $\delta u^t=0$). We have $u^\nu \nabla_\nu = \partial_t$ and $\Delta^\nu_i\nabla_\nu = \partial_i $ and the linearized evolution equations become:
\begin{align}
    \nabla_\mu \delta J^\mu &= n \partial_i \delta u^i + \partial_t \delta n +\partial_t \delta N + \partial
    _i \delta \mathcal{J}^i = 0, \label{eq:perturbationevolution1} \\\
   - u_\nu \nabla_\mu \delta T^{\mu\nu} &= \partial_t \delta \epsilon + (\epsilon+P)\partial_i \delta u^i + \partial_t \delta \mathcal{A} + \partial_i \delta Q^i = 0, \label{eq:perturbation_evolution2} \\
    \Delta^i_\nu \nabla_\mu \delta T^{\mu\nu} &= \epsilon\partial_t\delta u^i + \partial^i(\delta P+ \delta \Pi) + P\partial_t \delta u^i    - 2\eta\partial_j\delta\sigma^{ij} + \partial_t \delta Q^i = 0, \label{eq:perturbationevolution3} 
\end{align}
with the perturbations of the dissipative terms being
\begin{align}
    \delta \mathcal{A} &= \tau_\epsilon [\partial_t \delta \epsilon + (\epsilon + P)\partial_i \delta u^i] + \mathcal{O}(\delta\mathbf{q}^2), \\
    \delta \Pi &= - \zeta \partial_i \delta u^i + \tau_P [\partial_t \delta \epsilon + (\epsilon + P)\partial_i \delta u^i] + \mathcal{O}(\delta\mathbf{q}^2), \\
    \delta Q^i &= \tau_Q (\epsilon + P) \partial_t \delta u^i + \beta_\epsilon \partial^i \delta \epsilon + \beta_n \partial^i \delta n + \mathcal{O}(\delta\mathbf{q}^2), \\
    \delta N &= \tau_n (n \partial_i \delta u^i + \partial_t \delta n) + \mathcal{O}(\delta\mathbf{q}^2), \\
    \delta \mathcal{J}^i &= \tau_\mathcal{J} n \partial_t \delta u^i + \lambda_n \partial^i \delta n  + \lambda_\epsilon \partial^i \delta \epsilon +\mathcal{O}(\delta\mathbf{q}^2), \\
    \delta \sigma^{i j} &= \partial^{(i} \delta u^{j)} - \frac{1}{3}\delta^{ij}\partial_l \delta u^l + \mathcal{O}(\delta\mathbf{q}^2).
\end{align}

At this point we assume the perturbation to be a plane-wave (Fourier) mode in space with a exponential time dependence,

\begin{equation}
    \delta \boldsymbol{q} = e^{\Gamma t + i k^i x_i} \delta \boldsymbol{q}_0.
\end{equation}
Here $\boldsymbol{q}_0$ denote a constant amplitude vector, $k^i$ a spacelike wave vector, and $\Gamma$ the (generally complex) growth/decay rate.
Plugging this expressions into the linearized evolution equations and dropping terms of order $\mathcal{O}(\delta\mathbf{q}^2)$, we get a linear system $\mathbf{\mathcal{S}}(\Gamma,k^i)\delta \mathbf{q}=0$ such that
\begin{align}
    \mathbf{\mathcal{S}}(\Gamma,k^i) = \left[ \begin{array}{ccc}
            -\lambda_\epsilon k^2 & \tau_n\Gamma^2 + \Gamma - \lambda_n k^2 & in[1 + (\tau_n +\tau_\mathcal{J})\Gamma] k_j \\
            \tau_\epsilon \Gamma^2+  \Gamma - \beta_\epsilon k^2 & -\beta_n k^2 & i\rho[1 + (\tau_\epsilon+\tau_Q)\Gamma] k_j\\
            i \left[ \left. \frac{\partial P}{\partial\epsilon} \right|_n + (\tau_P+\beta_\epsilon) \Gamma \right] k^i & i\left[ \left. \frac{\partial P}{\partial n} \right|_\epsilon + \beta_n \Gamma \right] k^i & \mathcal{B}^i_j
        \end{array} \right],
\end{align}
where we have defined
\begin{align}
    \mathcal{B}^i_j &= \left[ \zeta + \frac{1}{3}\eta- \tau_P \rho \right] k^i k_j + [\rho \tau_Q \Gamma^2 + \rho \Gamma + \eta k^2] \delta^i_j. 
\end{align}

In order to find non trivial solutions of the perturbed equations, we need the determinant of $\mathbf{\mathcal{S}}(\Gamma,k^i)$ to vanish. Imposing $\textrm{det}[\mathbf{\mathcal{S}}(\Gamma,k^i)]=0$, will then allow to determine all the possible decay rates $\Gamma$. Sufficient condition for stability is that every non trivial solution must have $\Re(\Gamma)\leq0$ for any value of $k^i$. \\
The determinant of the stability matrix can be split into a shear and a sound channel as follows:
\begin{equation}
    \textrm{det} \left[ \mathcal{S}(\Gamma,k^i) \right] = -\left[P_{\textrm{shear}}(\Gamma,k)\right]^2P_{\textrm{sound}}(\Gamma,k),
\end{equation}
with $P_{\rm shear}$ and $P_{\rm sound}$ being two polynomials in $\Gamma$, whose roots will determine the stability of the shear and sound modes, respectively. 
They can be expressed as
\begin{align}
 P_{\textrm{shear}}(\Gamma,k) &= \rho \tau_Q \Gamma^2 + \rho \Gamma + \eta k^2, \label{eq:stability_polynomial_1} \\
 P_{\textrm{sound}}(\Gamma,k) &= a_6 \Gamma^6 + a_5 \Gamma^5 + a_4 \Gamma^4 + a_3 \Gamma^3 + a_2 \Gamma^2 + a_1 \Gamma + a_0, \label{eq:stability_polynomial_2}
\end{align}
where
\begin{align}
    a_6 &= \rho \tau_n \tau_Q \tau_\epsilon, \\
    a_5 &= \rho (\tau_n \tau_Q + \tau_n \tau_\epsilon + \tau_Q \tau_\epsilon), \\
    a_4 &= \rho (\tau_Q + \tau_n + \tau_\epsilon) + \left[ V \tau_n \tau_\epsilon + \tau_n \tau_\epsilon (n \beta_n + \rho \beta_\epsilon) + \rho \tau_Q (\tau_n \tau_P - \lambda_n \tau_\epsilon) + n \tau_\mathcal{J} \tau_\epsilon \beta_n \right] k^2, \\
    a_3 &=\rho + \left[ V(\tau_n+\tau_\epsilon) + \tau_n \tau_\epsilon \rho c_s^2 + n \left( \beta_n (\tau_\epsilon + \tau_n + \tau_\mathcal{J}) + \tau_\mathcal{J} \left. \frac{\partial P}{\partial n}\right|_\epsilon \tau_\epsilon \right) \right. \nonumber \\
    &+ \left. \rho \left( -\lambda_n (\tau_Q+\tau_\epsilon) + \left. \frac{\partial P}{\partial \epsilon} \right|_n \tau_Q \tau_n + \tau_Q \tau_P + \beta_\epsilon \tau_\epsilon \right)\right]k^2, \\
    a_2 &= \left[ V + n \left( \beta_n + \left. \frac{\partial P }{\partial n} \right|_\epsilon (\tau_n + \tau_\epsilon + \tau_\mathcal{J}) \right) + \rho \left( -\lambda_n + \left. \frac{\partial P}{\partial \epsilon}\right|_n (\tau_n + \tau_Q + \tau_\epsilon) \right) \right] k^2 \nonumber \\
    & + \left[ - V (\beta_\epsilon \tau_n + \lambda_n \tau_\epsilon) + n \beta_n ( \tau_\mathcal{J}+\tau_n) \tau_P + \rho \left( (\beta_n \lambda\epsilon - \beta_\epsilon \lambda_n) \tau_\epsilon + (\beta_\epsilon \tau_n-\lambda_n \tau_Q) \tau_P \right) \right] k^4, \\
    a_1 &= \rho c_s^2 k^2 + \left[ -V ( \lambda_n + \beta_\epsilon) + (\rho \beta_\epsilon + n \beta_n) \tau_P \right. \nonumber \\
    &+ \left. n (\tau_n + \tau_J) \left( \beta_n \left. \frac{\partial P}{\partial \epsilon} \right|_n - \beta_\epsilon \left. \frac{\partial P}{\partial n} \right|_\epsilon\right) + \rho (\tau_\epsilon + \tau_Q) \left(\lambda_\epsilon \left. \frac{\partial P}{\partial n} \right|_\epsilon -\lambda_n \left. \frac{\partial P}{\partial \epsilon} \right|_n \right) \right]k^4, \\
    a_0 &= \left[ n \left( \beta_n \left. \frac{\partial P}{\partial \epsilon} \right|_n - \beta_\epsilon \left. \frac{\partial P}{\partial n} \right|_\epsilon \right) - \rho \left( \lambda_n \left. \frac{\partial P}{\partial \epsilon} \right|_n - \lambda_\epsilon \left. \frac{\partial P}{\partial n} \right|_\epsilon \right)\right] k^4 + \left( \rho \tau_P - V \right) \left( \beta_n \lambda_\epsilon - \beta_\epsilon \lambda_n \right) k^6,
\end{align}
and with 
\begin{equation}
    c_s^2 = \left. \frac{\partial P}{\partial \epsilon} \right|_n + \frac{n}{\rho} \left. \frac{\partial P}{\partial n} \right|_\epsilon \label{eq:speed_of_sound}
\end{equation}
being the adiabatic speed of sound of the ideal fluid.
For determining the stability we employ the Routh-Hurwitz theorem \cite{Gantmacher1959v2, ogata2010modern} to the polynomials \eqref{eq:stability_polynomial_1} and \eqref{eq:stability_polynomial_2}, the Routh matrix of polynomial \eqref{eq:stability_polynomial_2} is
\begin{align}
   \boldsymbol{H} = \left[ \begin{array}{cccccc}
            a_5 & a_3 & a_1 & 0 & 0 & 0 \\
            a_6 & a_4 & a_2 & a_0 & 0 & 0 \\
            0 & a_5 & a_3 & a_1 & 0 & 0 \\
            0 & a_6 & a_4 & a_2 & a_0 & 0 \\
            0 & 0 & a_5 & a_3 & a_1 & 0 \\
            0 & 0 & a_6 & a_4 & a_2 & a_0 \\
        \end{array} \right].
\end{align}

According to the Routh-Hurwitz theorem, the stability for the polynomial \eqref{eq:stability_polynomial_2}, is guaranteed if $a_6>0$, and all the upper-left $n\times n$ submatrices have positive determinant, with $n$ going from 1 to 6. We define the matrices 
$\boldsymbol{H_1} = \left[a_5 \right]$, $\boldsymbol{H_2} = \left[ \begin{array}{cc}
            a_5 & a_3  \\
            a_6 & a_4 
        \end{array} \right]$,  $\boldsymbol{H_3} = \left[ \begin{array}{cccccc}
            a_5 & a_3 & a_1 \\
            a_6 & a_4 & a_2  \\
            0 & a_5 & a_3 
        \end{array} \right]$,..., $\boldsymbol{H_6}=\boldsymbol{H}$
and write the stability result as a Lemma (including the condition for polynomial \eqref{eq:stability_polynomial_1}).  

\begin{lemma} \label{stability}
Given the entropy conditions $\sigma+\sigma_0\geq0$ and $\eta\geq0$, the hyperbolicity condition $D\geq0$, the local thermodynamic stability requirement $s_2>0$, the assumptions $\tau_n,\tau_\epsilon,\tau_Q>0$, then the set of evolution equations \eqref{eq:evolution1}-\eqref{eq:evolution3}, with dissipative terms expressed as Eqs. \eqref{eq:A}-\eqref{eq:J}, are linearly stable around a background uniform solution if the following conditions are satisfied 
\begin{equation}
    0\leq\det(\boldsymbol{H_i})  \text{\ for \ }  i=2,3,4,5
\end{equation}
\end{lemma}

Note that $0\leq\det(\boldsymbol{H_1})$ and $0\leq\det(\boldsymbol{H_6})$ are required by the Routh-Hurwitz theorem, but not by the lemma, since they are automatically satisfied if the assumptions of the lemma hold. We explain this in the proof below.

\begin{proof}
    Linear stability of \eqref{eq:evolution1}-\eqref{eq:evolution3} requires all the roots in $\Gamma$ of the polynomials \eqref{eq:stability_polynomial_1} and \eqref{eq:stability_polynomial_2} to have non positive real part.  For \eqref{eq:stability_polynomial_1} this is trivially satisfied by the entropy condition $\eta\geq0$ and the assumption $\tau_Q>0$. This can be seen using the Routh-Hurwitz theorem \cite{Gantmacher1959v2, ogata2010modern}, which in case of a second order polynomial reduces to requiring all the coefficients to be non-negative.
    For \eqref{eq:stability_polynomial_2} the Routh-Hurwitz theorem states that sufficient conditions for stability are $0\leq\textrm{det}(\boldsymbol{H_i})$ for $i=1,2,..,6$, and $a_6>0$.
    Condition $a_6>0$ is trivially satisfied. Moreover, $det(\boldsymbol{H_6})=a_0det(\boldsymbol{H_5})$, implying that if $a_0\geq0$, the last condition for $i=6$ is redundant.
    It remains now to prove that $a_0\geq0$. This can be done substituting $\lambda_n$, $\lambda_\epsilon$, $\beta_n$, $\beta_\epsilon$ from equations  \eqref{lambda_1}-\eqref{beta_2} into the expression of $a_0$ to write
\begin{equation}
    a_0 = (\sigma+\sigma_0)(\epsilon+P) s_2 k^4 + D k^6
\end{equation}
 under the assumptions $\sigma+\sigma_0\geq0$ and $D\geq0$, and $s_2>0$ such expression is non-negative. 
Finally, we can get rid of the condition $det(\boldsymbol{H_1})\geq0$, since it is trivially satisfied by our assumptions $\tau_Q,\tau_n,\tau_\epsilon>0$.
    
\end{proof}

We linearize around a rest-frame homogeneous equilibrium in Minkowski spacetime for simplicity. The stability conditions obtained from this analysis are pointwise, namely, they depend only on the equation of state and on transport coefficients evaluated at the local fluid's configuration. In a general curved spacetime and for a generic background velocity, one can always work in a local inertial, comoving frame at each spacetime point. In that frame, the leading-order dynamics reduces to the Minkowski, rest-frame analysis we performed here.
For this reason, these stability conditions must be checked even when the background is not flat and the fluid is not at rest.

\subsection{Choice of frame} \label{subsection:choice_of_frame}

We want to find a frame that satisfies all the strong hyperbolicity, causality, entropy, and stability conditions for our formulation in the case where we don't have degenerate characteristic modes (i.e. $D\neq0$). We decide to adopt the frame class defined in \cite{Pandya:2022sff} and start by defining the following quantities, which allow us to express the conditions using only dimensionless variables:
\begin{align}
    &\kappa_\epsilon \equiv \frac{\rho^2 T}{n} \left. \frac{\partial(\mu/T)}{\partial \epsilon}\right|_n, \quad \kappa_n \equiv \rho T \left. \frac{\partial(\mu/T)}{\partial n} \right|_\epsilon, \quad \kappa_s \equiv \kappa_n + \kappa_\epsilon,
    \quad \alpha \equiv \left. \frac{1}{c_s^2} \frac{\partial P}{\partial \epsilon}\right|_n, \quad \omega \equiv \frac{\kappa_s}{\kappa_\epsilon}, \label{eq:EOS_parameters}
\end{align}
where $c_s^2$ is the adiabatic speed of sound defined in Eq. \eqref{eq:speed_of_sound}. With these definitions, we can express the transport coefficients defined in Eqs. \eqref{lambda_1}-\eqref{beta_2} as
\begin{align}
    \beta_e &= \tau_Q \alpha c_s^2 + \frac{\sigma}{\rho} \kappa_\epsilon, & \beta_n &= \tau_Q (1-\alpha)\frac{\rho}{n} c_s^2 + \frac{\sigma}{n} \kappa_\epsilon (\omega-1), \\
    \lambda_\epsilon &= \tau_\mathcal{J} \alpha c_s^2 \frac{n}{\rho} -\sigma_0\frac{n}{\rho^2} \kappa_\epsilon, & \lambda_n &= \tau_\mathcal{J} (1-\alpha) c_s^2 - \frac{\sigma_0}{\rho} \kappa_\epsilon (\omega-1).
\end{align}

The constitutive parameters are then factorized as:
\begin{align}
    &\tau_\epsilon = \tau_Q = \tau_n = \tau_\mathcal{J} \equiv \hat{\tau} L \hat{V}, \quad \tau_P \equiv 2 \alpha c_s^2 L \hat{V}, \nonumber \\
    &V \equiv \hat{V} L \rho c_s^2, \quad \eta \equiv\rho c_s^2 L \hat{\eta}, \quad \zeta \equiv \rho c_s^2 L \hat{\zeta}, \quad \hat{V}\equiv \hat{\zeta}+\frac{4}{3}\hat{\eta} \label{eq:frame_factorization} \\
    &\sigma \equiv \hat{\sigma} \hat{V} L \rho c_s^2 /(-\kappa_\epsilon), \quad \sigma_0 \equiv \tilde{\sigma} \hat{V} L \rho c_s^2 /(-\kappa_\epsilon) \nonumber 
\end{align}
with $\hat{\tau},\hat{\eta},\hat{\zeta},\hat{\sigma},\tilde{\sigma}$ being dimensionless parameters (that we are going to constraint), and $L$ a characteristic length. Assuming $\hat{\tau}>0$, and $\hat{\eta},\hat{\zeta},\hat{\sigma}+\tilde{\sigma}>0$, this choice automatically satisfies the causality and stability conditions of the shear channel, as well as the entropy conditions. We can express all the sound causality and stability conditions as functions of $\hat{\tau}$, $\hat{\sigma}$, $\tilde{\sigma}$, and the adimensional EoS variables $\alpha$, $\omega$, and $c_s^2$. We define the set of causality conditions we want to satisfy as:
\begin{align}
    C_1 &\equiv -A>0, & \hat{C}_1 &\equiv C_1/(L^3 \hat{\tau}^3 \hat{V}^3 \rho)\\
    C_2 &\equiv B>0,  & \hat{C}_2 &\equiv C_2 /(c_s^2 L^3 \hat{\tau}^2 \hat{V}^3 \rho)\\
    C_3 &\equiv -C>0, & \hat{C}_3 &\equiv C_3/(c_s^4 L^3 \hat{\tau} \hat{V}^3 \rho)\\
    C_4 &\equiv D>0, & \hat{C}_4 &\equiv C_4 /(c_s^6 L^3 \hat{\tau} \hat{V}^3 \rho)\\
    C_5 &\equiv -A-B-C-D>0, & \hat{C}_5 &\equiv C_5/(L^3 \hat{\tau} \hat{V}^3 \rho)\\
    C_6 &\equiv -3A-2B-C>0, & \hat{C}_6 &\equiv C_6/(L^3 \hat{\tau} \hat{V}^3 \rho)\\
    C_7 &\equiv -3A-B>0, & \hat{C}_7 &\equiv C_7/(c_s^4 L^3 \hat{\tau} \hat{V}^3 \rho)\\
    C_8 &\equiv 18 A B C D + B^2 C^2 - 27 A^2 D^2  - 4 A C^3 - 4 B^3 D>0, & \hat{C}_8 &\equiv C_8 /(c_s^{6} L^{12} \hat{\tau}^6 \hat{V}^{12} \rho^4)
\end{align}
where we additionally have defined the reduced causality conditions $\hat{C}_i>0$ in order to make the conditions dimensionless and remove the dependence on $L$ and $\hat{V}$.

On the other hand, we define the stability conditions as:
\begin{align}
    S_1 &\equiv \textrm{det} (\boldsymbol{H_2}) = S_1^{(0)} + S_1^{(2)} k^2\geq0 \\
        S_2 &\equiv \textrm{det} (\boldsymbol{H_3})  = S_2^{(0)} + S_2^{(2)}k^2 + S_2^{(4)} k^4 \geq 0 \\
        S_3 &\equiv \textrm{det} (\boldsymbol{H_4})  = S_3^{(2)} k^2 + S_3^{(4)} k^4 + S_3^{(6)} k^6 + S_3^{(8)} k^8 \geq0 \\
        S_4 &\equiv \textrm{det} (\boldsymbol{H_5})  = S_4^{(4)} k^4 + S_4^{(6)} k^6 + S_4^{(8)} k^8 + S_4^{(10)} k^{10}+ S_4^{(12)} k^{12} \geq 0 
\end{align}
Similarly to $\hat{C}_i$, and for the same reasons, we define the reduced stability conditions $\hat{S}_i^{(j)}$ by rescaling $S_i^{(j)}$:

\begin{align}
    \hat{S}_1^{(j)} &\equiv S_1^{(j)}/( c_s^jL^{3+j} \hat{\tau}^{3+j/2} \hat{V}^{3+j} \rho^2), & \hat{S}_2^{(j)} &\equiv S_2^{(j)}/(c_s^j L^{3+j}\hat{\tau}^{3+j/2}\hat{V}^{3+j} \rho^3)\\
    \hat{S}_3^{(j)} &\equiv S_3^{(j)}/( c_s^jL^{2+j} \hat{\tau}^{2+j/2} \hat{V}^{2+j} \rho^4), & \hat{S}_4^{(j)} &\equiv S_4^{(j)}/( c_s^jL^{j} \hat{\tau}^{1+j/2} \hat{V}^{j} \rho^5)
\end{align}

We emphasize that, as for $\hat{C}_i$, also $\hat{S}_i^j$ does not depend on $L$ and $\hat{V}$. Causality and stability requires $\hat{C}_i>0$ and $\hat{S}_i^{(j)}\geq0$, respectively, for every $i$ and $j$. The explicit form of the causality and stability conditions in term of the fluid variables and frame coefficients is extremely complicated and we do not report it here, but it can be found in the Mathematica notebooks attached \cite{notebooks_and_scripts}. Once we plug the frame parameters factorization into the expressions of the causality and stability quantities we obtain expressions of the type

\begin{align}
    \hat{C}_i &= \sum_{q=0}^{q_{max}} g_{(i,q)}(\hat{\sigma},\tilde{\sigma},\alpha,\omega,c_s^2)\hat{\tau}^q \\
    \hat{S}_i^{(j)} &= \sum_{q=0}^{q_{max}} h_{(i,j,q)}(\hat{\sigma},\tilde{\sigma},\alpha,\omega,c_s^2)\hat{\tau}^q
\end{align}
with $g_{(i,q)}$ and $h_{(i,j,q)}$ being smooth functions whose explicit expression can be found in the Mathematica notebooks attached, and $q_{max}$ being different for each equation. Following \cite{Pandya:2022sff}, we impose:
\begin{align}
    &\alpha>1, \quad 0<\omega<1/2, \quad 0<\alpha \omega<1/2, \quad 0<c_s^2<1,  \label{eq:eos_constraints}\\
    &0\leq\tilde{\sigma}<1/3, \quad 0 < \hat{\sigma} + \tilde{\sigma} < 1/3, \quad \hat{\tau}>0. \label{eq:frame_constraints}
\end{align}
We can show that $\hat{C}_1$, $\hat{C}_2$, $\hat{C}_4$ are always positive under these assumptions, as well as $\hat{S}_1^{(0)}$, $\hat{S}_1^{(2)}$, $\hat{S}_2^{(0)}$, $\hat{S}_2^{(2)}$, $\hat{S}_3^{(2)}$, $\hat{S}_4^{(4)}$. For the remaining conditions, we can show that the previous assumptions guarantee that

\begin{align}
    g_{(i,q_{max})}(\hat{\sigma},\tilde{\sigma},\alpha,\omega,c_s^2)&>0 \\
    h_{(i,j,q_{max})}(\hat{\sigma},\tilde{\sigma},\alpha,\omega,c_s^2)&>0
\end{align}
hence, every $\hat{C}_i$ and $\hat{S}_i^{(j)}$ will tend to infinity for $\hat{\tau}\rightarrow+\infty$, guaranteeing that there exists a $\hat{\tau}$ large enough to satisfy all the causality and stability conditions for a given set of $\alpha$, $\omega$, $c_s^2$ $\hat{\sigma}$, $\tilde{\sigma}$.

\begin{figure}
    \centering
    \includegraphics[width=0.45\linewidth]{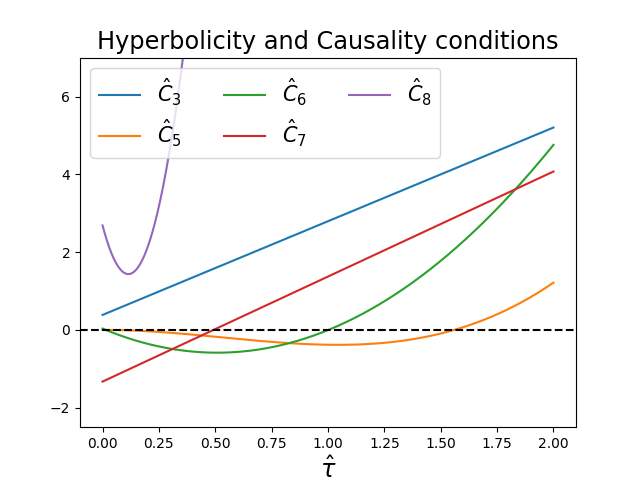}
    \includegraphics[width=0.45\linewidth]{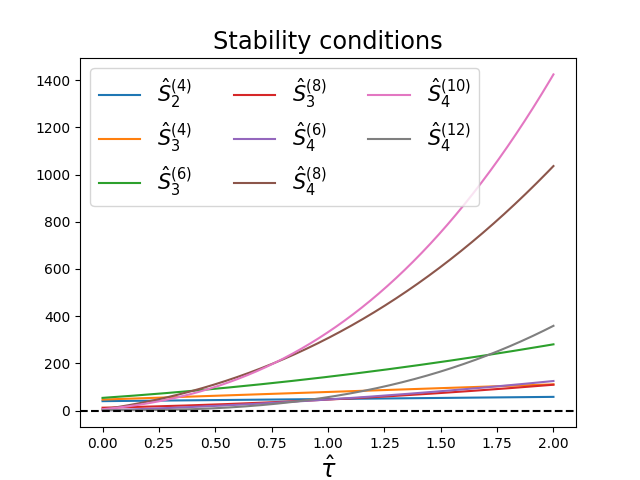}
    \caption{Hyperbolicity, Causality, and stability conditions as function of $\hat{\tau}$ for $\alpha=5/3$, $c_s^2=0.3$, $\omega=0.1$, $\hat{\sigma}=-0.05$, $\tilde{\sigma}=0.15$.}
    \label{fig:coditions}
\end{figure}

There are still two caveats to mention before concluding our analysis. First, it is the fact that
\begin{equation}
    g_{(5,q_{\rm max})} = 1 - c_s^2
\end{equation}
implying that the leading coefficient in $\hat{\tau}$ of $\hat{C}_5$ tends to vanish when $c_s^2\rightarrow1$. This means the $\hat{\tau}$ can become arbitrarily large when $c_s^2$ approaches the speed of light. This issue has been found also in \cite{Pandya:2022sff} and makes this choice of frame not suitable for fluids where the speed of sound can get close to luminal speed. \\
The second caveat comes from the non existence of an upper bound for $\alpha$, this variable can indeed diverge as $\alpha \sim 1/c_s^2$ when $c_s^2\rightarrow0$, requiring arbitrarly large $\hat{\tau}$ as in the previous case. However, this is not a source of concern. Indeed, in the Mathematica notebooks attached, one can see that the leading term of $\tilde{C}_i$ and $\tilde{S}_i^{(j)}$ in $1/c_s^2$ always has positive coefficient under the conditions of Eqs. \eqref{eq:eos_constraints} and \eqref{eq:frame_constraints}, leading the conditions to be always satisfied in the limit of vanishing $c_s^2$. This can be seen writing the causality and stability conditions expressions as series in $1/c_s^2$, and noticing that in this limit $\alpha\sim1/c_s^2$ and, in order to keep $\alpha \omega$ bound, $\omega\sim c_s^2$. This guarantees the existence of a finite $\hat{\tau}$ satisfying all causality and stability conditions in the limit of vanishing $c_s^2$ (or infinite $\alpha$).

In FIG. (\ref{fig:coditions}) we show an example of the ideas just described. Both the causality and stability conditions are plotted as a function of $\hat{\tau}$ for a fixed value of all other parameters. As one can see all quantities grow monotonically (beyond a certain point), becoming eventually  positive. This allows to find a value of $\hat{\tau}$ big enough to satisfy all the conditions.

Finally, we provide the reader with a set of explicit values of $\hat{\sigma}$, $\tilde{\sigma}$, and $\hat{\tau}$ that satisfy all the conditions on a range of hydrodynamics parameters often encountered in simulations. 

\begin{claim} \label{claim}
Given the evolution equations \eqref{eq:evolution1}, \eqref{eq:evolution2}, \eqref{eq:evolution3} with $u^\mu u_\mu=-1$, dissipative terms given by Eqs. \eqref{eq:A}-\eqref{eq:J},
and an equation of state laying in the three-dimensional parameter space
\begin{equation}
    \left\{ \left. \frac{\partial P}{\partial \epsilon} \right|_n, \alpha\omega, c_s^2 \right\} \in [10^{-6},1] \times[0,1/2]\times[10^{-6},0.9] \cap \{\alpha>1\}.
    \label{eq:eos_range}
\end{equation}
Then, we claim that any frame defined by the parametrization of Eq. \eqref{eq:frame_factorization} with dimensionless coefficients satisfying the conditions 
\begin{equation}
    0\leq\tilde{\sigma} < \frac{1}{3}, \quad 0<\tilde{\sigma}+\hat{\sigma}<\frac{1}{3}, \quad \hat{\tau}=\frac{2+\delta}{1-c_s^2}, \quad \delta>0, \quad \hat{\zeta}\geq0, \quad \hat{\eta}>0\label{eq:frame_constraints}
\end{equation}
yields a strongly hyperbolic, causal, and linearly stable theory with positive entropy production at $\mathcal{O}(\nabla^3)$.
\end{claim}

This statement can be tested numerically using the Python script attached \cite{notebooks_and_scripts}. The test consists in sampling the parameter space \eqref{eq:eos_range} with a three-dimensional $n\times n\times n$ grid removing the points where $\alpha\leq1$, and computing $\hat{C}_i$ and $\hat{S_i}$ in every point. The grid points has been distributed uniformly in $\alpha \omega$ and logarithmically in $\left. \frac{\partial P}{\partial \epsilon} \right|_n$ and $c_s^2$. We find that all the conditions are satisfied on the sampled points and the results do not significantly change increasing the grid resolution up $n=300$. Given the smoothness of $\hat{C}_i$ and $\hat{S}^{(j)}_i$, we claim it is reasonable to assume that these conditions are satisfied on the whole domain of \eqref{eq:eos_range}. Additionally, we compute the squared characteristic velocities $c_{s_i}^2$ of the sound channel using Eq. \eqref{eq:sound_velocities}, and we verify they are always real and causal. 

We highlight that in case of an ideal gas EoS, $\partial P/\partial\epsilon|_n = \Gamma-1$, $\alpha \omega = (mn/\epsilon)(\Gamma-1)/\Gamma$, with $m$ being the particle's rest mass, and $\omega=(mn/\epsilon)P/(\epsilon+P)<1/2$ (since $P<\epsilon$), $c_s^2=\Gamma P/(\epsilon+P)<\Gamma/2$ see  \cite{Pandya:2022sff} for more details. Hence, the range \eqref{eq:eos_range} covers any ideal gas EoS with $\Gamma\in(1,1.8)$. In appendix \ref{appendix:hybrid_eos}, we show that this frame choice can also be applied to piecewise polytropic EoS with thermal contribution in the form of an ideal gas, with non very restrictive conditions on the EoS parameters, and in a range of densities and energies usually encountered in binary neutron star mergers.

\section{Coupling to Einstein field equations}
\label{section:coupling_with_gravity}
Dealing with first-order hydrodynamics in general relativity brings an additional difficulty with respect to the ideal fluid case: the hydrodynamics equations are now second order as the Einstein field equations (EFE). This means we have to study the strong hyperbolicity of the coupled system employing a principal matrix pencil containing both fluid and spacetime evolution equations. The strong hyperbolicity of the mixed EFE-BDNK system has already been proven straightforwardly in \cite{Bemfica:2017wps}, \cite{Bemfica:2019knx}, \cite{Bemfica:2020zjp}. In this section we show that, for our formulation, this can be done easily at second order using the strong hyperbolicity definition of \cite{Abalos:2026kim} as long as the fluid and spacetime eigenvalues are not degenerate. 

We want to solve the previous system together with the EFE:
\begin{align}
    G_{\mu\nu} &= 8\pi T_{\mu\nu}
\end{align}

\begin{lemma} \label{SABDNK-GR}
    Consider the Einstein field equations in generalized harmonic gauge, coupled with the evolution equations \eqref{eq:evolution1}-\eqref{eq:evolution3}, with dissipative terms expressed as Eqs. \eqref{eq:A}-\eqref{eq:J}. If the hyperbolicity and causality conditions of Lemma \ref{Lemma:hip:cau} are satisfied, the resulting system is strongly hyperbolic and causal.
\end{lemma}

\begin{proof}
The perturbation vector and principal symbol, using the generalized harmonic formulation, can now be expressed as

\begin{align}
    \delta \boldsymbol{\psi} &\equiv \left[ \begin{array}{c}
           \delta n \\
           \delta \epsilon \\
           \delta u^{\nu} \\ \hline
           \delta g^{\mu\nu}
        \end{array} \right] \equiv \left[ \begin{array}{c}
           \delta \mathbf{f} \\ \hline
           \delta \mathbf{g}
        \end{array} \right], \quad
        \boldsymbol{\mathcal{M}}^{\alpha\beta}l_\alpha l_\beta \equiv \left[ \begin{array}{c|c}
            \mathfrak{N}^{\alpha\beta}l_{\alpha}l_{\beta} & \boldsymbol{b}^{\alpha\beta}l_\alpha l_\beta \\ \hline
            \boldsymbol{O}_{10\times5} & \boldsymbol{I}_{10}g^{\alpha\beta}l_{\alpha}l_{\beta} 
        \end{array} \right]
\end{align}
with $\mathfrak{N}^{\alpha\beta}l_{\alpha}l_{\beta}$ being the principal symbol of the hydrodynamics equations defined in Eq. \eqref{eq:principal_symbol}, $\boldsymbol{O}_{10\times5}$ is a $10\times5$ zero matrix, $\boldsymbol{I}_{10}$ is the identity matrix with dimension 10, and the coupling $\boldsymbol{b}^{\alpha\beta}l_\alpha l_\beta$ between EFE and the hydrodynamics sector is a $10\times5$ matrix whose lines are expressed as

\begin{align}
    \boldsymbol{b}^{\alpha\beta}l_\alpha l_\beta & \equiv \left[ \begin{array}{c}
           B_1^{\lambda\sigma} \\
           B_2^{\lambda\sigma} \\
           B^{\gamma\lambda\sigma}_3\\
        \end{array} \right]
\end{align}

with

\begin{align}
   B_1^{\lambda \sigma} &\equiv - \frac{1}{2}n\tau_n (u\cdot l)^2 g^{\lambda \sigma} - \rho \tau_{\mathcal{J}} \left[ (u\cdot l) u^\lambda l^\sigma - \frac{1}{2}l^2 u^\lambda u^\sigma \right]\\
   B_2^{\lambda \sigma} & \equiv \frac{1}{2} \rho \tau_\epsilon (u \cdot l)^2 g^{\lambda \sigma} + \rho \tau_Q \left[(u\cdot l) u^\lambda l^\sigma-\frac{1}{2} l^2 u^\lambda u^\sigma \right]\\
   B_3^{\gamma\lambda \sigma} & \equiv \frac{1}{2} \left[ \rho \tau_P - \zeta\right] (u\cdot l) \Delta^{\gamma\alpha}l_\alpha g^{\lambda\sigma} - \eta (u\cdot l) \Delta^{\lambda\sigma \gamma \mu} l_\mu \nonumber \\
   & + \rho \tau_Q \left[ (u\cdot l)^2 u^\lambda \Delta^{\gamma \sigma} - \frac{1}{2}(u\cdot l) \Delta^\gamma_{\alpha}l^\alpha u^\lambda u^\sigma \right]
\end{align}

and $\Delta^{\alpha\beta\mu\nu}\equiv \frac{1}{2}(\Delta^{\alpha\mu} \Delta^{\beta\nu} + \Delta^{\mu\beta}\Delta^{\nu\alpha}-\frac{2}{3}\Delta^{\mu\nu}\Delta^{\alpha\beta})$ so that $\sigma^{\mu\nu}=\Delta^{\mu\nu\alpha\beta}\nabla_\alpha u_\beta$.
The determinant of the principal symbol is then unaffected by $\boldsymbol{b}^{\alpha\beta}l_\alpha l_\beta$, leaving the characteristic velocities of both hydrodynamics and spacetime sectors unchanged. Considering now, the same 3+1 foliation as for  the fluid sector, we obtain the same eigenvalues $\pm c_{sh}, \pm c_{s_i}$ for the fluid and $\pm1$ for the EFE. The latter have degeneracy 10. The coupled system is then at least weakly hyperbolic, and causal. It remains now to compute the kernel vectors to check for strong hyperbolicity. The kernel vectors  can be affected by the presence of $\boldsymbol{b}^{\alpha\beta}l_\alpha l_\beta$. 
This implies we will need to take into account the effect of the fluid-spacetime coupling in the principal matrix when writing the element of the kernel for the full system. We have to compute the elements of the kernel for each characteristic mode requiring for:

\begin{equation}
    \boldsymbol{\mathcal{M}}^{\alpha\beta}l_\alpha l_\beta \delta \boldsymbol{\psi} = \left[ \begin{array}{c}
         \mathfrak{N}^{\alpha\beta}l_{\alpha}l_{\beta}  \delta \mathbf{f} + \boldsymbol{b}^{\alpha\beta}l_\alpha l_\beta\delta \mathbf{g} \\ \hline
        \boldsymbol{I}_{10}g^{\alpha\beta} l_\alpha l_\beta \delta \mathbf{g}
    \end{array} \right] = 0 \label{eq:kernel_eq_full}
\end{equation}

We already know that the fluid and EFE sectors are strongly hyperbolic when taken separately. So, let us assume that $\delta \boldsymbol{f}$ belongs to the kernel of the fluid sector, and $\delta \boldsymbol{g}$ belongs to the EFE kernel.
It is straightforward to verify that the kernel vectors of the coupled system are:

\begin{equation}
    \delta \boldsymbol{\psi} = \left[ \begin{array}{c}
        \delta \mathbf{f} \\ \hline
        0
    \end{array} \right] \label{eq:kernel_eq_full} \quad \textrm{for hydrodynamics modes,}
\end{equation}

\begin{equation}
    \delta \boldsymbol{\psi} = \left[ \begin{array}{c}
        -\left( \mathfrak{N}^{\alpha\beta}l_{\alpha}l_{\beta} \right)^{-1} \left( \boldsymbol{b}^{\alpha\beta}l_\alpha l_\beta \right) \delta \mathbf{g} \\ \hline
        \delta \mathbf{g}
    \end{array} \right] \label{eq:kernel_eq_full} \quad \textrm{for EFE modes.}
\end{equation}

So every element of the hydrodynamics and EFE kernels, can be mapped into one element of the full system's kernel as long as $\mathfrak{N}^{\alpha\beta}l_{\alpha}l_{\beta}$ is invertible for EFE characteristic modes, i.e. as long as the hydrodynamics characteristic modes are not degenerate with the EFE ones, this requirement is guaranteed by the causality conditions. Moreover, the map depends smoothly on $k^\mu$. Hence the strong hyperbolicity of the system is preserved as long as the hydrodynamics characteristic mode are not degenerate with the EFE ones.
From a geometrical point of view this is equivalent to say that the null cones of the effective metric in the fluid sector, do not touch the null cones of spacetime metric.
This is imposed asking the inequalities coming from the causality conditions to hold in strong form. \\
\end{proof}

\section{Summary}

We extend the first-order general-relativistic viscous hydrodynamics model of \cite{Bemfica:2020zjp} to include a fully first-order charge current in the same spirit used there for the stress--energy tensor: namely, by adding out-of-equilibrium contributions proportional to the ideal-fluid equations of motion. Moreover, in contrast to \cite{Bemfica:2020zjp}, we work algebraically imposing the constraint $u_\mu u^\mu=-1$, proving  well posedness in the whole ambient Sobolev space of the evolved variables. This is the approach typically adopted for models intended for numerical applications. Concretely, our model is summarized by Eqs.~\eqref{eq:A}-\eqref{eq:J} (together with the transport-coefficient identifications in Eqs.~\eqref{lambda_1}-\eqref{beta_2}, yielding a closed set of five second-order evolution equations for $(\epsilon,n,u^\mu)$ that can be written in flux-conservative form at the level of the fundamental conservation laws.

Using the matrix-pencil criterion for second-order systems introduced in \cite{Abalos:2026kim} \footnote{We comment about the criterion presented in~\cite{Abalos:2026kim}, the failure of them does not exclude the possibility that another,
non-standard, reduction could be strongly hyperbolic (there is an example about this at
the end of section 6.1). However, this is not sufficient by itself to establish well-posedness of
the original system. One must still show that the energy estimate obtained for the reduced variables
can be translated back into the expected Sobolev estimate for the original variables, without
loss of derivatives. This additional step is non-trivial. In particular, a strongly hyperbolic
reduced system may give an estimate in variables that are not uniformly equivalent to the
natural variables of the original PDE system.}, we derive sufficient conditions for strong hyperbolicity and causality (Lemma~\ref{Lemma:hip:cau}) directly at second order, without performing an explicit first-order reduction. However, we emphasize that the criteria of~\cite{Abalos:2026kim} are fully equivalent to the standard first-order strong-hyperbolicity criteria, while avoiding the need to perform an explicit first-order reduction. We also establish that the canonical entropy current \eqref{eq:entropy current} produces non-negative entropy generation up to $\mathcal{O}(\nabla^3)$ provided $\zeta\ge 0$, $\eta\ge 0$ and $\sigma+\sigma_0\ge0$ (Lemma~\ref{Lemma:entropy}), and we obtain pointwise sufficient conditions for linear stability around homogeneous equilibrium via a Routh--Hurwitz analysis of the sound-channel polynomial (Lemma~\ref{stability}).

Because the full set of causality and stability restrictions is algebraically cumbersome, we adopt a physically motivated class of frames (following \cite{Pandya:2022sff}) and exhibit a broad parameter range in which all requirements are simultaneously satisfied. In particular, Claim~\ref{claim} provides an explicit family of frame parameters (Eq.~\eqref{eq:frame_factorization} with the restrictions \eqref{eq:frame_constraints}) that yields a strongly hyperbolic, causal, and linearly stable theory with positive entropy production for a wide region of equations of state characterized by the parameter domain \eqref{eq:eos_range} and the additional restrictions \eqref{eq:eos_constraints}. This identifies a practical region of parameter space suitable for numerical applications with realistic equations of state beyond the simplified uncharged setups. We supplement this section with two Mathematica notebooks based on the one introduced in \cite{Pandya:2022sff}, and a Python script. Both can be found in \cite{notebooks_and_scripts}.

Finally, we study the coupling of our charged BDNK system to the Einstein field equations in the harmonic gauge. We show that, under the fluid-sector hyperbolicity and causality conditions, the coupled second-order principal symbol remains strongly hyperbolic and causal (Lemma~\ref{SABDNK-GR}), provided the fluid characteristic speeds do not coincide with the gravitational ones. Equivalently, the strong-hyperbolicity property is preserved as long as the null cones of the effective metrics governing the fluid characteristic modes remain strictly inside the spacetime light cone. Under some extra conditions, we can establish that our charged formulation can be consistently embedded in fully dynamical spacetimes for different gauge fixing for numerical relativity implementations.

As a physics-motivated application, Appendix~\ref{appendix:hybrid_eos} shows that a piecewise-polytropic equation of state with an ideal-gas thermal contribution satisfies the assumptions of Lemma~\ref{stability} and can realize the hypotheses of Claim~\ref{claim} in regimes typical of binary neutron star mergers.

Overall, the results presented here define a well-posed, causal, and thermodynamically consistent charged first-order dissipative theory in the BDNK framework, suitable for direct numerical implementation. Natural next steps include: (i) implementing the formulation in a flux-conservative first-order reduction (as in~\cite{Shum:2025jnl,Clarisse:2025lli}) to robustly handle discontinuities, (ii) exploring optimized frame selections adapted to specific EoS, and (iii) performing systematic tests in relativistic-astrophysics scenarios where charge and diffusion effects are expected to be relevant.

\section{Acknowledgments}

We thank Harry Shum for carefully reading the manuscript and for providing helpful comments and suggestions.
We thank Carlos Palenzuela for helpful discussions and comments. We thank Giorgio Torrieri for helpful comments on the manuscript.
This work was partially supported by the project PID2022-138963NB-I00, funded by the Spanish Ministry of
Science, Innovation and Universities (MCIN/AEI/10.13039/501100011033). 
We acknowledge the use of an AI-assisted writing tool for copy-editing and improving the presentation of the manuscript; all scientific results, derivations, and conclusions were developed by the authors.

\bibliography{ref}

\appendix
\section{Hyperbolicity without imposing 4-velocity normalization}
\label{appendix:non_normalized_u_hyperbolicity}

The conservation equations $\nabla_\mu T^{\mu\nu}=0$ provide four
equations for $\epsilon$ and the normalized velocity $u^\mu$.
Following the spirit of Ref.~\cite{Bemfica:2020zjp}, we may instead treat the four
components of $u^\mu$ as independent variables and adopt the velocity
evolution equations \eqref{eq:evolution3} as four independent
equations, imposing normalization on the initial data. Together with
the equation for $\epsilon$, this introduces one additional equation.
Establishing propagation of the normalization constraint is essential:
without it, equivalence with the normalized formulation cannot be
guaranteed. We show this below.

Defining $C=u^\mu u_\mu+1$, contracting Eq. \eqref{eq:evolution3} with $u_\mu$ yields
\begin{equation}
\label{eq:normalization_constraint_propagation}
\mathfrak{N}^{\mu}\nabla_\mu C + C L = 0,
\end{equation}
where
\begin{equation}
\begin{aligned}
\mathfrak{N}^{\mu}
&=
\left[
\epsilon+\mathcal A+P+\Pi
+\tau_Q(\epsilon+P)\nabla_\nu u^\nu
\right]u^\mu+Q^\mu,
\\
L
&=2\Bigl[
u^\lambda\nabla_\lambda(P+\Pi)
-u_\nu\nabla_\lambda(2\eta\sigma^{\lambda\nu})
+u_\nu u^\lambda\nabla_\lambda Q^\nu
%\\
%&
%\qquad
+(\nabla_\mu u^\mu)
\left(
\beta_\epsilon u^\lambda\nabla_\lambda\epsilon
+\beta_n u^\lambda\nabla_\lambda n
\right)
\Bigr].
\end{aligned}
\end{equation}

Assuming that the coefficient of the time derivative of $C$ remains
nonzero throughout the evolution, zero initial constraint violation
remains zero at later times.

Equations \eqref{eq:evolution1}-\eqref{eq:evolution3} now form a system of six second-order equations for the six variables \((n,\epsilon,u^\beta)\), which we show below to be strongly hyperbolic for arbitrary initial data, including data that violate the constraints. This property is essential for numerical evolutions of this formulation, where truncation and round-off errors inevitably generate constraint violations.

Lets now consider the principal matrix of Eq. \eqref{eq:principal_symbol} without imposing $u^\mu \delta u_\mu = 0$, as presented in \cite{Bemfica:2020zjp}. The contraction of Eq. \eqref{eq:evolution2} with $u^\mu$ now does not identically vanish anymore, and the perturbation vector can be expressed as:

\begin{equation}
    \delta u^\nu = a k^\nu + b u^{\nu} + c e_1^\nu + d e_2^\nu
\end{equation}

with $e_i^\nu u_\nu = e_i^\nu k_\nu = 0$ and $e_i^\nu e_{j,\nu}=\delta_{ij}$. The principal matrix of the perturbation expressed in this base is now:

\begin{equation}
    \textrm{det} \left( \mathfrak{N}^{\mu\nu}l_{\mu}l_{\nu} \right) = -(\rho \tau_Q \lambda^2 - \eta)^3 (Ax^3 + Bx^2 + C x + D)
\end{equation} 

with respect to the previous case we have now a degeneracy 3 for the shear mode $\lambda^2=\eta/\rho \tau_Q$. \\
The principal matrix in this basis is:

\begin{align}
   \mathfrak{N}^{\mu\nu}l_{\mu}l_{\nu} = \left[ \begin{array}{ccc|ccc}
             \lambda_\epsilon& \tau_n \lambda^2 + \lambda_n  & n(\tau_n + \tau_\mathcal{J}) \lambda & \tau_n n \lambda^2 & 0 & 0\\
            \tau_\epsilon \lambda^2 + \beta_{\epsilon} & \beta_n & \rho(\tau_\epsilon+\tau_Q) \lambda  & 0 & 0 & 0 \\
            (\tau_P+\beta_\epsilon) \lambda & \beta_n \lambda & \left( \rho \tau_P - \zeta - \frac{4}{3} \eta \right) + \rho \tau_Q \lambda^2 & \rho \tau_P - \zeta - \frac{4}{3} \eta  & 0 & 0 \\ \hline
            0 & 0 & 0 & -(\rho \tau_Q \lambda^2 - \eta) & 0 & 0 \\
            0 & 0 & 0 & 0 & \rho \tau_Q \lambda^2 - \eta & 0 \\
            0 & 0 & 0 & 0 & 0 &  \rho \tau_Q \lambda^2 - \eta  \\
        \end{array} \right]
\end{align}

It is easy to verify that the previous kernel vectors $\boldsymbol{r}^{(\textrm{sound})}$ remain unchanged, as the 4-th line always implies $b=0$. \\
Concerning $\boldsymbol{r^{(\textrm{shear})}}$, as $\rho \tau_Q \lambda^2 - \eta$=0, we have now 3 independent kernel vectors, the previous 2, plus the following one:

\begin{equation}
    \boldsymbol{r}^{(i)} = \left[ \begin{array}{c}
            F \\
            G \\
            u^\nu + a v^\nu
        \end{array} \right]
\end{equation}

with $a$, $F$, $G$ given by:

\begin{equation}
    \left[ \begin{array}{c}
            F \\
            G \\
            a
        \end{array} \right] = \left[ \begin{array}{ccc}
             \lambda_\epsilon v^2& \tau_n \lambda^2 + \lambda_n v^2 & n(\tau_n + \tau_\mathcal{J}) \lambda v^2 \\
            \tau_\epsilon \lambda^2 + \beta_{\epsilon} v^2 & \beta_n v^2 & \rho(\tau_\epsilon+\tau_Q) \lambda v^2 \\
            (\tau_P+\beta_\epsilon)\lambda v^2 & \beta_n \lambda v^2 & \left( \rho \tau_P - \zeta - \frac{4}{3} \eta \right) v^2 + \rho \tau_Q \lambda^2 v^2 
        \end{array} \right] ^{-1} \left[ \begin{array}{c}
            \tau_n n \lambda^2 \\
            0 \\
            \rho \tau_P - \zeta - \frac{4}{3} \eta
        \end{array} \right]
\end{equation}

which only admits one solution for $F$, $G$ and $a$ if the determinant of the upper left block of the principal matrix does not vanish, which is the case for the shear modes.
This means the system remains strongly hyperbolic also without imposing the normalization of $u^\mu$, the only consequence is the appearance of a new (unphysical) mode.

\section{The case $\beta_n=0$}
\label{appendix:betaniszero}
In this appendix we are going to spend some words on the case $\beta_n=0$. Such case, despite not being of particular physical interest, has the appealing mathematical property of making charge and energy-momentum weekly coupled in the principal part. The principal symbol can indeed be written as:

\begin{equation}
\mathfrak{N}^{\alpha\beta}l_{\alpha}l_{\beta}=\left[
\begin{array}
[c]{ccc}%
\lambda_{\epsilon}\Delta^{\alpha\beta} & \tau_{n}u^{\alpha}u^{\beta}%
+\lambda_{n}\Delta^{\alpha\beta} & n(\tau_{n}+\tau_{\mathcal{J}})\delta_{\nu
}^{(\alpha}u^{\beta)}\\
\tau_{\epsilon}u^{\alpha}u^{\beta}+\beta_{\epsilon}\Delta^{\alpha\beta} &
0 & \rho(\tau_{\epsilon}+\tau_{Q})u^{(\alpha
}\delta_{\nu}^{\beta)}\\
(\tau_{P}+\beta_{\epsilon})u^{(\alpha}\Delta^{\beta)\mu} & 0 & C_{\nu}^{\mu\alpha\beta}%
\end{array}
\right]  l_{\alpha}l_{\beta},%
\end{equation}
and the eigenvalue problem becomes the product of three independent eigenvalue problems
\begin{equation}
       \textrm{det}\left( \mathfrak{N}^{\mu\nu}l_{\mu}l_{\nu} \right) = \left( \rho \tau_Q \lambda ^2 - \eta \right)^2 (\tau_n \lambda^2 + \lambda_n) (A x^2 + B x + C) = 0
\end{equation}
with $x=\lambda^2$ and

\begin{align}
    A &= \rho \tau_\epsilon \tau_Q \\
    B &= - \tau_\epsilon \left( \beta_\epsilon \rho  + \zeta+\frac{4}{3}\eta \right) - \rho \tau_P \tau_Q \\
    C &= \beta_\epsilon \left(\rho \tau_P - \zeta - \frac{4}{3}\eta\right)
\end{align}

Notice the $A$, $B$ and $C$ do not depend in any way from the charge transport coefficients, and they reduce to those defined in \cite{Bemfica:2020zjp} for $\beta_n=0$. In addition to the sound channel velocities of \cite{Bemfica:2020zjp} with $\beta_n=0$, we have the usual shear mode whose velocity remains unchanged, and a mode associated with charge diffusion with velocity $\lambda= \pm \sqrt{-\lambda_n/ \tau_n}$. This particular choice constitutes a useful testbed for testing numerical code aimed at solving the evolution equations \eqref{eq:evolution1}-\eqref{eq:evolution3}, as it keeps the characteristics of the well tested model of \cite{Bemfica:2020zjp} unchanged.

It remains now the question whether it is possible to set $\beta_n=0$ without breaking any of the hyperbolicity, causality, entropy, and stability conditions. Using the definition of $\beta_n$ in Eq. \eqref{beta_2}, and the parameters defined in Eq. \eqref{eq:EOS_parameters}, the condition $\beta_n=0$ can be expressed as
\begin{equation}
    \hat{\sigma} (1-\omega) = \hat{\tau} (\alpha-1)
\end{equation}
substituting the definitions of $\alpha$ and $c_s^2$, and choosing $\hat{\tau}=(2+\delta)/(1-c_s^2)$ with $\delta>0$ to satisfy the frame condition of Eq. \eqref{eq:frame_constraints}, one can isolate $\hat{\sigma}$ and express it as
\begin{equation}
    \hat{\sigma} = \frac{1}{1-\omega} \frac{2+\delta}{1-p'_\epsilon + \nu} \frac{\nu}{p'_\epsilon - \nu} \leq \frac{1}{3} \label{eq:condition_for_sigma}
\end{equation}
where we have defined
\begin{equation}
    p'_\epsilon = \left. \frac{\partial P}{\partial \epsilon} \right|_n, \quad \nu = - \frac{n}{\rho} \left. \frac{\partial P}{\partial n} \right|\epsilon,
\end{equation}
and we imposed the last inequality in order to satisfy the frame conditions \eqref{eq:frame_constraints} for $\hat{\sigma}$. It is clear that, Eq. \eqref{eq:condition_for_sigma} admits a solution only when $\nu$ is small enough. We can give clear physical meaning to this result in the case of the ideal gas EoS. In this case we have \cite{Pandya:2022sff}
\begin{equation}
    \omega = \frac{mnP}{\epsilon \rho}, \quad p'_\epsilon = \Gamma-1, \quad \nu = \frac{mn}{\rho} (\Gamma-1)
\end{equation}
and the condition $\nu \ll 1$ translates into $\rho=\epsilon+P \gg mn$, i.e. the gas should be in a relativistic regime. This is reasonable as in this regime the rest-mass energy contribution brought by the charge density becomes negligible in the energy budget, hence the system's evolution becomes independent on the charge density.

\section{positivity of $s_2$}
\label{appendix:positivity_of_detJ(e,n)}

In chapter 8 of \cite{callen1993thermodynamics}  the local thermodynamic stability is motivated by requiring that the entropy be strictly concave, so that equilibrium is a local maximum of entropy. In our formulation the state is parametrized by $(\epsilon,n)$, and the same concavity requirement becomes strict concavity of the entropy density $s(\epsilon,n)$. In this appendix we show that this concavity condition is equivalent to $s_2>0$.

For this demonstration we will make use of the first principle of thermodynamics in differential form
\begin{equation}
    d\epsilon = T ds + \mu dn, \label{eq:first_principle}
\end{equation}
of the Gibbs-Duhem relation for the pressure
\begin{equation}
    dP = s dT + n d\mu, \label{eq:Gidd_Duhem}
\end{equation}
and  
\begin{equation}
    sT = \epsilon +P-\mu n,
\end{equation}
obtained by the combination of the previous two equations.

Starting from the first principle of thermodynamics Eq. \eqref{eq:first_principle}, we can write
\begin{equation}
    s_\epsilon = \frac{1}{T}, \quad s_n = - \frac{\mu}{T} \label{eq:dsde_dsdn}
\end{equation}
where the subscripts mean partial derivatives with respect to a certain variable. In particular, the second expression of Eq. \eqref{eq:dsde_dsdn}, implies
\begin{equation}    
\left. \frac{\partial(\mu/T)}{\partial \epsilon} \right|_n = - s_{n\epsilon}, \quad \left. \frac{\partial(\mu/T)}{\partial n} \right|_\epsilon = - s_{nn}, \label{eq:dmude_dmudn}
\end{equation}
Now, using the Gibbs-Duhem relation \eqref{eq:Gidd_Duhem} we obtain
\begin{equation}
    \left. \frac{\partial P}{\partial n} \right|_\epsilon = s T_n + n \mu_n, \quad \left. \frac{\partial P}{\partial \epsilon} \right|_n = s T_\epsilon + n \mu_\epsilon \label{eq:dPdn_dPde}
\end{equation}
Using Eqs. \eqref{eq:dsde_dsdn} and \eqref{eq:dmude_dmudn}, one can write
\begin{align}
    T_\epsilon = - T^2s_{\epsilon\epsilon}, \quad \mu_\epsilon = -T s_{\epsilon n} - \mu T s_{\epsilon\epsilon} \\
        T_n = - T^2s_{\epsilon n}, \quad \mu_n = -T s_{n n} - \mu T s_{\epsilon n} 
\end{align}
which substituted into Eq. \eqref{eq:dPdn_dPde}, gives
\begin{equation}
    \left. \frac{\partial P}{\partial n} \right|_\epsilon = -T^2 (s-ns_n) s_{\epsilon n} - n Ts_{n n}, \quad  \left. \frac{\partial P}{\partial \epsilon} \right|_n = -T^2 (s-ns_n) s_{\epsilon\epsilon} - n Ts_{\epsilon n}. \label{eq:dPdn_dPde_in_s}
\end{equation}

Finally, combining Eqs. \eqref{eq:dPdn_dPde_in_s} and \eqref{eq:dmude_dmudn}, we obtain
\begin{equation}
    T\left(\left. \frac{\partial (\mu/T)}{\partial n} \right|_\epsilon \left. \frac{\partial P}{\partial \epsilon} \right|_n - \left. \frac{\partial (\mu/T)}{\partial \epsilon} \right|_n \left. \frac{\partial P}{\partial n} \right|_\epsilon \right) = \frac{T(s-ns_n)(s_{nn}s_{\epsilon\epsilon}-s_{n\epsilon}^2)}{s_{\epsilon}^2} > 0,
\end{equation}
since
\begin{equation}
    T(s-ns_n) = T s + \mu n = \epsilon+P > 0,
\end{equation}
and the concavity of the entropy implies
\begin{equation}
    s_{nn}s_{\epsilon\epsilon}-s_{n\epsilon}^2 = \textrm{det}\left(\frac{\partial^2 s}{\partial^2 (\epsilon, n )} \right) >0. 
\end{equation}

Concavity of the entropy is  a property that must be satisfied by every stable thermodynamical system. For this reason we assume the EoS we are employing to satisfy this property as well.

\section{Piecewise polytropic EoS with thermal contribution}
\label{appendix:hybrid_eos}
In this appendix, we show that the piecewise polytropic EoS with ideal-gas thermal contribution, as introduced in \cite{Shibata:2005ss}, can be employed in the formulation we investigated in this article, since it satisfies all the conditions of Lemma \ref{stability} (under some simple assumptions on the EoS). We show this by computing the explicit expression of $\left. \partial P / \partial\epsilon \right|_n$, $\alpha$ and $\omega$ and verifying that they can satisfy the hypothesis of Claim \ref{claim} for a reasonable choice of EoS parameters and in a range of $n$ and $\epsilon$ usually encountered in binary neutron star merger simulations. This type of EoS, indeed, is still widely used in neutron star simulations, and it was the first EoS to introduce thermal effect in neutron star mergers simulations. Despite not being the current state of the art for simulations involving detailed microphysics, it remains an important milestone towards the simulation of realistic neutron star mergers in BDNK theory.

We start by defining the pressure as a cold contribution depending on the charge density only (in this case the baryon charge) and a thermal contribution depending on the energy density as well: 
\begin{equation}
    P(n,\epsilon) = P_{cold}(n) + P_{th}(n,\epsilon).
\end{equation}

For the cold contribution we use a piecewise-polytropic fit, so that
\begin{equation}
    P_{\rm cold}(n) = K_i\, \rho_b^{\Gamma_i}.
\end{equation}
Here the baryon mass density is defined as $\rho_b \equiv m_b n$, where $n$ is the baryon number density and $m_b$ is a reference baryon mass. The index $i$ labels the density interval $\rho_b \in (\rho_{b,i}, \rho_{b,i+1})$; the constants $K_i$ and $\Gamma_i$ are constant within each interval, but take different values in different intervals.
$\Gamma_i$ is obtained from a fit of nuclear cold EoS, while $K_i$ is set to impose continuity (with the exception of $K_1$), see \cite{Read:2008iy} for more details.
The thermal contribution is modeled with an ideal-gas EoS,
\begin{equation}
    P_{th}(n,\epsilon) = (\Gamma_{th}-1)m_b n e_{th}(n,\epsilon) \equiv n T \label{eq:ideal_gas_pressure}
\end{equation}
where $\Gamma_{\rm th}$ is the thermal adiabatic constant and $e_{\rm th}(n,\epsilon)$ is the thermal internal energy per baryon in units of the baryon rest mass.

To compute $e_{\rm th}$, we split the total energy density into cold and thermal parts,
\begin{equation}
    \epsilon = \rho_b \bigl[e_{\rm cold}(n) + e_{\rm th}(n,\epsilon)\bigr],
\end{equation}
with $\rho_b = m_b n$. The cold contribution $e_{\rm cold}(n)$ includes both the rest-mass energy and the fermion degeneracy energy, and is constructed using the first law of thermodynamics at zero temperature,
\begin{equation}
    d e_{\rm cold}(n) = - P_{\rm cold}(n)\, d\!\left(\frac{1}{\rho_b}\right). \label{eq:e_cold}
\end{equation}

For the piecewise-polytropic fit in the interval $\rho_b \in (\rho_{b,i},\rho_{b,i+1})$, this yields
\begin{equation}
    e_{\rm cold}(n)
    = a_i + \frac{K_i}{\Gamma_i-1}\,\rho_b^{\Gamma_i-1},
\end{equation}
where the constants $a_i$ are chosen to enforce continuity across intervals. In this convention, $a_i$ also contains the rest-mass contribution, so that $a_i\gtrsim1$.

Finally, the thermal energy per baryon is obtained from the above decomposition as
\begin{equation}
    e_{\rm th}(n,\epsilon) = \frac{\epsilon}{\rho_b} - e_{\rm cold}(n).
\end{equation}

\begin{lemma}
    Consider a piecewise-polytropic EoS parametrized in the i-th density interval by the constants $\Gamma_i$, $K_i$ and $a_i$, together with a thermal contribution modeled as an ideal gas, as given in Eq. \eqref{eq:ideal_gas_pressure} and parametrized by $\Gamma_{\rm th}$. Assume that the standard ideal fluid causality condition $P < \epsilon$ holds, and that $c_s^2$ is limited by Eq. \eqref{eq:eos_range}, with an upper bound $c_{s,\rm max}^2$. Then, the system of equations \eqref{eq:evolution1}-\eqref{eq:evolution3} with the constitutive relations \eqref{eq:A}-\eqref{eq:J} is strongly hyperbolic, causal, and exhibits non-negative entropy production up to $\mathcal{O}(\nabla^3)$, provided that the following conditions are satisfied:
\begin{equation}
    e_{th}(n,\epsilon)<\frac{a_i}{K_i\rho_b^{\Gamma_i}}+ \frac{\Gamma_i}{\Gamma_i-1} \frac{1}{\rho_b}. %\implies
     %\epsilon < \frac{a_i}{K_i\rho_b^{\Gamma_i-1}} + \frac{\Gamma_i}{\Gamma_i-1} + \rho_b a_i + \frac{K_i}{\Gamma_i-1} \rho_b^{\Gamma_i}
     \label{eq:constraint_omegabigger0}
\end{equation}
\begin{equation}
    2c_{s,\rm max}^2\geq\Gamma_{th}> 1 + \frac{K_i\Gamma_i\rho_b^{\Gamma_i-1}}{a_i + K_i\frac{\Gamma_i}{\Gamma_i-1}\rho_b^{\Gamma_i-1}}
    \label{eq:constraint_alphabigger1}
\end{equation}

\end{lemma}

The combination of these three constraints can be easily satisfied in the regimes encountered in neutron star mergers. In geometrized units, indeed, the nuclear saturation density is of the order of $\rho_b\sim10^{-3}$, and the EoS parameters lie in typical ranges of $K_i\lesssim10^4$, $1<\Gamma_i \lesssim 3$ (and $a_i\gtrsim1$ to include the baryon rest mass energy). This makes the right hand side of Eq. \eqref{eq:constraint_alphabigger1} of the form $1 + \mathcal{O}(10^{-1})$ at maximum for most of the EoS in the literature (see \cite{Read:2008iy} for a summary), making the constraint on $\Gamma_{th}$ not very restrictive. Finally, the condition \eqref{eq:constraint_omegabigger0} is also not very restrictive, as the right hand side can reach minima of $\sim 10^{3}$ in the high density regime, allowing for an extremely high amount of thermal energy. Such high energy regime is never reached in any binary neutron star merger simulation. 

The assumption that the speed of sound remains within the range of Eq.~\eqref{eq:eos_range} (as well as the second inequality of Eq. \eqref{eq:constraint_alphabigger1}) may break down at extremely high densities due to the contribution of the cold component. This typically occurs when the fitting formula for $P_{\rm cold}$ is extrapolated beyond its range of validity. This issue can be remedied by introducing an additional density interval in which $\Gamma_i$ is calibrated so as to enforce causality.

\begin{proof}
We start by computing the derivatives of pressure with respect to $n$ and $\epsilon$, we notice that
\begin{equation} \frac{\partial e_{th}}{\partial \epsilon} = \frac{1}{m_b n}, \quad \frac{\partial
    e_{th}}{\partial n} = - \frac{\epsilon}{m_b n^2} - m_b\frac{d e_{cold}}{d \rho_b}
\end{equation}
from which we can easily get the pressure derivatives as
\begin{align}
    \left. \frac{\partial P}{\partial \epsilon} \right|_n &= \Gamma_{th}-1 \label{eq:dPdepsilon_pwp}\\
    \left. \frac{\partial P}{\partial n} \right|_\epsilon &= m_b K_i \Gamma_i \rho_b^{\Gamma_i-1} - (\Gamma_{th}-1)m_b \left[ a_i + K_i\frac{\Gamma_i}{\Gamma_i-1} \rho_b ^{\Gamma_i-1}\right]  \label{eq:dPdn_pwp}
\end{align}

These results match those of \cite{Pandya:2022sff} when the EoS reduces to an ideal gas, i.e. when $a_i=1$ and $K_i=0$. Following the conditions found \eqref{eq:eos_range}, we obtain that $\Gamma_{th}\in(1,2]$, as for the ideal gas case.

%%%%%%%%%%%%%%%%%%5
Using Eqs. \eqref{eq:dPdepsilon_pwp} and \eqref{eq:dPdn_pwp}, the sound speed can be exressed as
\begin{equation}
c_s^2=\frac{\Gamma_i P_{\rm cold}+\Gamma_{\rm th}P_{\rm th}}{\epsilon+P}\, .
\end{equation}
In the ultra-relativistic regime, where the thermal contribution dominates so that $P\simeq P_{\rm th}$, this reduces to
\begin{equation}
c_s^2 \simeq \Gamma_{\rm th}\,\frac{P}{\epsilon+P}\, .
\end{equation}
For physically admissible states one has $P <\epsilon $, implying
$\frac{P}{\epsilon+P}\le \frac{1}{2}$.
Hence $c_s^2$ is bounded by $c_s^2\le \Gamma_{\rm th}/2$. A sufficient condition to keep $c_s^2$ within the range of Eq.~\eqref{eq:eos_range} is therefore
\begin{equation}
1<\Gamma_{\rm th}<2\,c_{s,\max}^2\, .
\end{equation}

For computing the entropy and the baryon chemical potential, we follow \cite{Pandya:2022sff}, employing the first principle of thermodynamics written in the following form
\begin{equation}
    de = T d \left( \frac{s}{\rho_b} \right) - P d \left( \frac{1}{\rho_b} \right), \quad de = de_{cold} + de_{th}.
\end{equation}
Using $de_{cold}$ from Eq. \eqref{eq:e_cold}, we can get rid of the cold component and rewrite this as 
\begin{equation}
    de_{th} = T d \overline{s} - P_{th} d \left( \frac{1}{\rho_b} \right)
\end{equation}
with $\overline{s}=s/\rho_b$ being the entropy per baryon mass. Using the expression of $P_{th}(n,\epsilon)$, and expanding the differential $d \left( \frac{1}{\rho_b} \right)$, one can obtain the following differential relation
\begin{equation}
    m_bd\overline{s} = \frac{1}{\Gamma_{th}-1} \frac{de_{th}}{e_{th}} - \frac{dn}{n}
\end{equation}
that can be integrated to obtain

\begin{equation}
    s = n \left[ \frac{1}{\Gamma_{th}-1} \ln \left( \frac{e_{th}(n,\epsilon)}{n^{\Gamma_{th}-1}} \right) + const. \right].
\end{equation}
This expression is the same as in \cite{Pandya:2022sff}, since the entropy is only determined by the thermal part of the EoS, which is the same in both cases.

Once we have computed the entropy, we can use the identity

\begin{equation}
    n \mu = \epsilon + P - Ts
\end{equation}

to obtain the baryon chemical potential. Plugging the expressions for $T$, $s$ and $P$ we find

\begin{equation}
    \frac{\mu}{T} = \frac{1}{(\Gamma_{th}-1)e_{th}(n,\epsilon)} \left( a_i + K_i \frac{\Gamma_i}{\Gamma_i-1} \rho_b^{\Gamma_i-1} \right) + \frac{1}{\Gamma_{th}-1}\left[ \Gamma_{th} - \ln \left( \frac{e_{th}(n,\epsilon)}{n^{\Gamma_{th}-1}} \right) + const. \right] \label{eq:muoverT}
\end{equation}

We can now take derivatives of Eq. \eqref{eq:muoverT} and obtain the parameters

\begin{align}
    \kappa_\epsilon &= -\frac{\rho^2}{n^2} \left[ \frac{1}{e_{th}(n,\epsilon)} \left( a_i + K_i \frac{\Gamma_i}{\Gamma_i-1} \rho_b^{\Gamma_i-1}\right) + 1 \right],\\
    \kappa_n &= \frac{\rho}{n^2}\frac{\epsilon + P_{cold}(n)}{e_{th}(n,\epsilon)} \left( a_i + K_i\frac{\Gamma_i}{\Gamma_i-1} \rho_b^{\Gamma_i-1} \right) + \frac{2 \rho}{n^2}P_{cold}(n) + \frac{\rho \epsilon}{n^2} + \rho \frac{P_{th}(n,\epsilon)}{n^2},
\end{align}
where again we defined $\rho=\epsilon+P$, and their sum
\begin{equation}
     \kappa_s = \kappa_\epsilon + \kappa_n =-\frac{\rho}{n^2} (\Gamma_{th}-1)\rho_b\left( a_i + K_i\frac{\Gamma_i}{\Gamma_i-1} \rho_b^{\Gamma_i-1} \right) + \frac{\rho}{n^2}P_{cold}(n).
\end{equation}

The quantities $\kappa_n$, $\kappa_\epsilon$, $\kappa_s$ match those of \cite{Pandya:2022sff} for $K_i=0$ and $a_i=1$, i.e. when the cold contribution vanishes and we get back to the ideal gas EoS.
 Finally we can write the condition

\begin{equation}
    0<\omega = \frac{\kappa_s}{\kappa_\epsilon} = \frac{P_{th}(n,\epsilon)}{\rho} \times \frac{(a_i+K_i\frac{\Gamma_i}{\Gamma_i-1}\rho_b^{\Gamma_i-1})-e_{th}(n,\epsilon)P_{cold}(n)}{(a_i+K_i\frac{\Gamma_i}{\Gamma_i-1}\rho_b^{\Gamma_i-1}) + e_{th}(n,\epsilon)} < \frac{1}{2} 
\end{equation}
the condition $\omega<1/2$ is automatically satisfied by the causality condition $P<\epsilon$. We only need to impose $\omega>0$, which translates into
\begin{equation}
    e_{th}(n,\epsilon)<\frac{a_i}{K_i\rho_b^{\Gamma_i}}+ \frac{\Gamma_i}{\Gamma_i-1} \frac{1}{\rho_b}. %\implies
     %\epsilon < \frac{a_i}{K_i\rho_b^{\Gamma_i-1}} + \frac{\Gamma_i}{\Gamma_i-1} + \rho_b a_i + \frac{K_i}{\Gamma_i-1} \rho_b^{\Gamma_i}
\end{equation}

Then we need to impose the condition
\begin{align}
    &\alpha=\frac{1}{c_s^2}\left.\frac{\partial P}{\partial \epsilon}\right|_n = \frac{\left. \frac{\partial P}{ \partial \epsilon} \right|_n}{\left. \frac{\partial P}{ \partial \epsilon} \right|_n + \frac{n}{\rho}\left. \frac{\partial P}{ \partial n} \right|_\epsilon} > 1,
\end{align}
this is equivalent to imposing $\left. \partial P / \partial n \right|_\epsilon<0$, which, using Eq. \eqref{eq:dPdn_pwp}, can be expressed as

\begin{equation}
    \Gamma_{th}-1 >  \frac{K_i\Gamma_i\rho_b^{\Gamma_i-1}}{a_i + K_i\frac{\Gamma_i}{\Gamma_i-1}\rho_b^{\Gamma_i-1}}
    \label{eq:Gamma_th_bound}
\end{equation}

The last condition we need to impose is

\begin{align}
    \alpha \omega = \frac{P_{th}(n,\epsilon)}{\rho + \frac{n}{\Gamma_{th}-1}\left. \frac{\partial P}{\partial n}\right|_\epsilon} \times \frac{(a_i+K_i\frac{\Gamma_i}{\Gamma_i-1}\rho_b^{\Gamma_i-1})-e_{th}(n,\epsilon)P_{cold}(n)}{(a_i+K_i\frac{\Gamma_i}{\Gamma_i-1}\rho_b^{\Gamma_i-1}) + e_{th}(n,\epsilon)} < \frac{1}{2}
\end{align}

the second factor of this expression is already positive and smaller than 1, so we only need to impose the first factor to be smaller than $1/2$, which translates into

\begin{align}
    P_{th}(n,\epsilon) < \epsilon_{th} + \frac{K_i\Gamma_i\rho_b^{\Gamma_i}}{\Gamma_{th}-1}
    \label{eq:constraint_alphaomega_smallehalf}
\end{align}

which is always satisfied since $P_{th}<\epsilon_{th}$.  

Finally we can confirm that 

\begin{equation}
   s_2 = T\left( \left. \frac{\partial P}{\partial \epsilon} \right|_n  \left. \frac{\partial (\mu/T)}{\partial n} \right|_\epsilon - \left. \frac{\partial P}{\partial n} \right|_\epsilon  \left. \frac{\partial (\mu/T)}{\partial \epsilon} \right|_n \right) = -\kappa_\epsilon \frac{c_s^2}{\rho} \left( 1 - \alpha \omega\right) > 0 \quad \textrm{(since $\kappa_\epsilon<0$)},
\end{equation}
as required by hyperbolicity and thermodynamic stability conditions, and assumed throughout the article.

In summary, the only three conditions that need to be imposed for having strong hyperbolicity, causality and stability are the upper bound on $e_{th}(n,\epsilon)$ given by Eq. \eqref{eq:constraint_omegabigger0}, the lower bound on $\Gamma_{th}$ given by Eq. \eqref{eq:Gamma_th_bound}, and the requirement $\Gamma_{th}\in(1,2c_{s,\rm max}^2)$ coming from Eq. \eqref{eq:dPdepsilon_pwp} and Eq. \eqref{eq:eos_range}. The last two together giving Eq. \eqref{eq:constraint_alphabigger1}
\end{proof}

\end{document}